\documentclass[preprint]{my-elsarticle}

\usepackage{lineno,hyperref}
\usepackage{amsmath}
\usepackage{amssymb,amsfonts,amsmath,mathtools,amsthm,enumerate}
\usepackage{float}
\usepackage{graphicx}
\usepackage{amssymb}
\usepackage{resizegather}
\usepackage{longtable}
\usepackage{url}
\usepackage{enumerate}
\usepackage{multirow,cases,subfigure}

\newtheorem{theorem}{Theorem}
\newtheorem{lemma}[theorem]{Lemma}

\journal{\rm This is a preprint whose final form is published 
in 'Chaos, Solitons \& Fractals' 
at [\url{https://doi.org/10.1016/j.chaos.2024.115633}].}

\begin{document}

\begin{frontmatter}

\title{Parameters estimation and uncertainty assessment in the transmission dynamics of rabies in humans and dogs}

\author[mysecondaryaddress,mymainaddress]{Mfano Charles}\corref{mycorrespondingauthor}
\cortext[mycorrespondingauthor]{Corresponding author}
\ead{mfanoc@nm-aist.ac.tz}

\author[comymainaddress]{Sayoki G. Mfinanga}
\ead{gsmfinanga@yahoo.com}

\author[mysecondaryaddress]{G.A. Lyakurwa}
\ead{geminpeter.lyakurwa@nm-aist.ac.tz}

\author[pt]{Delfim F. M. Torres}
\ead{delfim@ua.pt}

\author[mysecondaryaddress]{Verdiana G. Masanja}
\ead{verdiana.masanja@nm-aist.ac.tz}

\address[mysecondaryaddress]{School of Computational and Communication Science and Engineering,\\ 
Nelson Mandela African Institution of Science and Technology (NM-AIST), P.O. Box 447, Arusha, Tanzania}

\address[mymainaddress]{Department of ICT and Mathematics, 
College of Business Education (CBE), P.O. BOX 1968 Dar es Salaam, Tanzania}

\address[comymainaddress]{NIMR Chief Research Scientist Fellow}

\address[pt]{Center for Research and Development in Mathematics and Applications (CIDMA),\\
Department of Mathematics, University of Aveiro, 3810-193 Aveiro, Portugal}

% ---------------------------------------

\begin{abstract}
Rabies remains a pressing global public health issue, demanding effective modeling and control strategies. 
This study focused on developing a mathematical model using ordinary differential equations (ODEs) 
to estimate parameters and assess uncertainties related to the transmission dynamics of rabies in humans and dogs. 
To determine model parameters and address uncertainties, next-generation matrices were utilized to calculate 
the basic reproduction number ${\cal R}_{0}$. Furthermore, the Partial Rank Correlation Coefficient was used 
to identify parameters that significantly influence model outputs. The analysis of equilibrium states revealed 
that the rabies-free equilibrium is globally asymptotically stable when $R_0<1$, whereas the endemic equilibrium 
is globally asymptotically stable when ${\cal R}_{0}\geq 1$. To reduce the severity of rabies and align with 
the Global Rabies Control (GRC) initiative by 2030, the study recommends implementing 
control strategies targeting indoor domestic dogs.
\end{abstract}

\begin{keyword}
Rabies dynamics \sep  Mathematical Modeling\sep Latin Hypercube Sampling (LHS) \sep Contact rate \sep Parameter uncertainty.
\end{keyword}
\end{frontmatter}

% ---------------------------------------

\section{Introduction}

Rabies is a fatal and neglected viral disease that affects mammals, including humans. 
The disease is primarily spread through bites and scratches by infected animals 
\cite{bonilla2022mapping,barecha2017epidemiology,kabeto2021rabies,kavoosian2023comparison,lembo2010feasibility}. 
The disease is caused by  rabies virus, a member of the family Rhabdoviridae and the genus Lyssavirus. 
The disease typically presents with symptoms such as inflammation, fever, headache, and malaise. 
These initial symptoms can progress to more specific manifestations including anxiety, confusion, hallucinations, 
dysphagia, muscle spasms, and paralysis, therefore, immediate medical attention 
is essential in the event of a bite or scratch from animals such as dogs and cats 
\cite{rocha2017epidemiological,taylor2015global,marsden2006rabies,monath2013vaccines}. 
Prevention of rabies primarily involves vaccinating humans and domestic dogs, 
and avoiding contact with wild animals that are  carriers of rabies virus 
\cite{barecha2017epidemiology,maki2017oral}. In many countries, vaccination programs 
have been successful in reducing the incidence of rabies in humans and domestic dogs 
\cite{velasco2017successful,gossner2020prevention}. Post-exposure prophylaxis (PEP) 
is a vaccine recommended for individuals who have been bitten or scratched by an animal 
who are  carrier of rabies virus. PEP comprises a series of injections with both rabies 
vaccine and rabies immune globulin to prevent the virus from causing an infection 
\cite{kabeto2021rabies}. 

Mathematical models are fundamental tools for comprehending 
and addressing the transmission dynamics of rabies. They aid in describing and predicting 
epidemiological phenomena, allocating resources, and structuring control strategies. 
Researchers in different countries have developed several mathematical models to 
study the transmission dynamics of rabies in dog populations 
\cite{abdulmajid2021analysis,amoako2021rabies,abrahamian2022rhabdovirus,ruan2017modeling,ayoade2019saturated,chapwanya2022environment}, 
as well as interactions between dog and human populations \cite{kadowaki2018risk, ayoade2023modeling, tulu2017impact}, 
and even among dog, human, and other animal populations \cite{pantha2021modeling,eze2020mathematical,ega2015sensitivity}. 
These studies have identified several factors that can shape the dynamics of rabies in different countries 
and ways of controlling the disease. However, due to the complex nature of disease transmission dynamics 
and the inherent variability in real-world data, parameter estimation poses significant challenges.
Estimating these parameters accurately is essential for designing effective control measures and interventions 
to prevent the spread of rabies. One common approach to parameter estimation involves using the least 
square method and validating the model by fitting real-world data to the proposed model. It is crucial 
to assess uncertainty in parameter estimates to understand the reliability of model predictions 
and make informed decisions about disease control strategies. Uncertainty can arise from various sources, 
including variability in data, model structure, and parameter estimation techniques.

To address this issue, statistical techniques such as Latin Hypercube Sampling (LHS) 
and Partial Rank Correlation Coefficient (PRCC) are utilized for uncertainty and sensitivity 
analysis of input parameters \cite{renardy2019global}. LHS generates random samples within 
the parameter space, while PRCC assesses the correlation between model output and each input 
parameter \cite{renardy2019global,gomero2012latin}. LHS enables efficient sampling of multi-dimensional 
parameter space, while PRCC computes and ranks the correlation coefficient between the model inputs 
(parameters) and the model outputs (state variables). It gauges the degree of a non-linear and 
monotonic relationship between model inputs (parameters) and model outputs. Therefore, this  
paper discusses the transmission dynamics of rabies, the parameter estimation, and the role of 
statistical techniques such as LHS and PRCC in evaluating the impact of uncertainties on model predictions. 

The structure of this paper is organized with sections dedicated to formulating the 
mathematical model (Section~\ref{modelFormulation}), qualitative analysis (Section~\ref{sec:02}), 
quantitative analysis (Section~\ref{sec:03}), and discussion and concluding the findings 
(Sections~\ref{sec:04} and \ref{sec:05}, respectively).

% ---------------------------------------------  
 
\section{Model Formulation} 
\label{modelFormulation}

The model for rabies transmission is divided into four distinctive components. 
The first component pertains to the human population, which is represented by $N_{H}=S_H+E_H+R_H+I_H$. 
The second component is the environmental virus reservoir, represented by $M$. The third component 
is comprised of free-range dogs, represented by $N_{F}=S_F+E_F+I_F$, which includes both stray and feral dogs. 
The fourth component represents domestic dogs, represented by $N_{D}=S_D+E_D+R_D+I_D$. Susceptible humans 
$\left(S_H\right)$ are individuals who can be infected. They are recruited into the population 
at a rate of $\theta_{1}$. Exposure happens upon contact with infected free-range dogs $\left(I_{F}\right)$, 
domestic dogs $\left(I_{D}\right)$, or the virus $\left(M\right)$ at rates of $\tau_{1}$, $\tau_{2}$, or $\tau_{3}$, 
respectively. The latency period for exposed individuals is 1-3 months. Thus the force of infection for human  is given by
\begin{equation}
    \chi_{1} = (\tau_{1} I_{F} + \tau_{2} I_{D} + \tau_{3} \lambda(M))S_{H},
\end{equation}
where
\begin{align*}
\lambda(M) = \frac{M}{M+C}.
\end{align*}
Exposed individuals who receive post-exposure prophylaxis recover at a rate of $\beta_{2}$. 
Vaccinated individuals lose their protective immunity and become susceptible  at a rate of $\beta_{3}$. 
In some cases, exposed individuals vaccinated with less efficious  vaccine may progress 
to infectious state $I_{H}$ at the transition rate of  $\beta_{1}$. Infected individuals may  
suffer to induced deaths at a rate of $\sigma_{1}$. Susceptible free-range dogs $S_F$ become 
infected upon contact with infected individuals $I_{F}$, $I_{D}$, or the rabies virus in the 
environment at rates of $\kappa_{1}$, $\kappa_{2}$, or $\kappa_{3}$, respectively. 
The  force of  infection of free range dogs is  given by
\begin{equation}
\chi_{2} =\left(\kappa_{1}I_{F}+\kappa_{2} I_{D}+\kappa_{3} \lambda \left(M\right)\right)S_{F}.
\end{equation}
Upon exposure to rabies, susceptible free-range dogs transition to the latent state $E_{F}$, 
where they remain for 1 to 3 months. Subsequently, they advance to the infectious state $I_{F}$ 
at a rate $\gamma$. Infected free-range dogs face mortality at a rate $\sigma_{2}$, 
while all free-range dogs experience natural death  at a rate of $\mu_{2}$.

Susceptible domestic dogs ($S_{D}$) are continually recruited at a rate $\theta_{3}$ and become infected 
through contact with either infected individuals ($I_{F}$,\; $I_{D}$,\; $M$) or the virus in the environment 
at rates $\psi_{1}$, $\psi_{2}$, or $\psi_{3}$, respectively, and the force of  infection is defined as
\begin{equation}
\chi_{3}=\left(\dfrac{\psi_{1}I_{F}}{1+\rho_{1}}+\dfrac{\psi_{2}I_{D}}{1+\rho_{2}}
+\dfrac{\psi_{3}}{1+\rho_{3}}\lambda \left(M\right)\right) S_{D}.
\end{equation}
Following exposure to rabies, susceptible domestic dogs transition to the latent state $E_{D}$, 
where they remain for a period of time at a rate $\beta_{1}$. Those in $E_{D}$ who receive 
post-exposure prophylaxis move to the recovered state $R_{D}$ at rate $\gamma_{2}$. However, 
as post-exposure prophylaxis does not provide permanent immunity, individuals in $R_{D}$ 
can lose immunity and become susceptible again at rate $\gamma_{3}$. The remaining portion 
of $E_{D}$ progresses to the infectious state $I_{D}$ at rate $\gamma_{1}$. Infected domestic 
dogs face mortality due to the disease at rate $\sigma_{3}$. Additionally, all domestic dogs 
experience natural death at a rate of $\mu_{3}$. The rabies virus in the environment 
is introduced through shedding from infectious free-range dogs, domestic dogs, and humans 
at rates $\nu_{2}$, $\nu_{3}$, and $\nu_{1}$, respectively, at the general rate  defined as
\begin{equation}
\theta_{4}=\left(\nu_1I_H+\nu_2I_F+\nu_3I_D\right)M.
\end{equation}	
The viruses are eliminated from the environment at a rate of $\mu_{4}$.
\begin{figure}[H]
\centering
\includegraphics[scale=0.55]{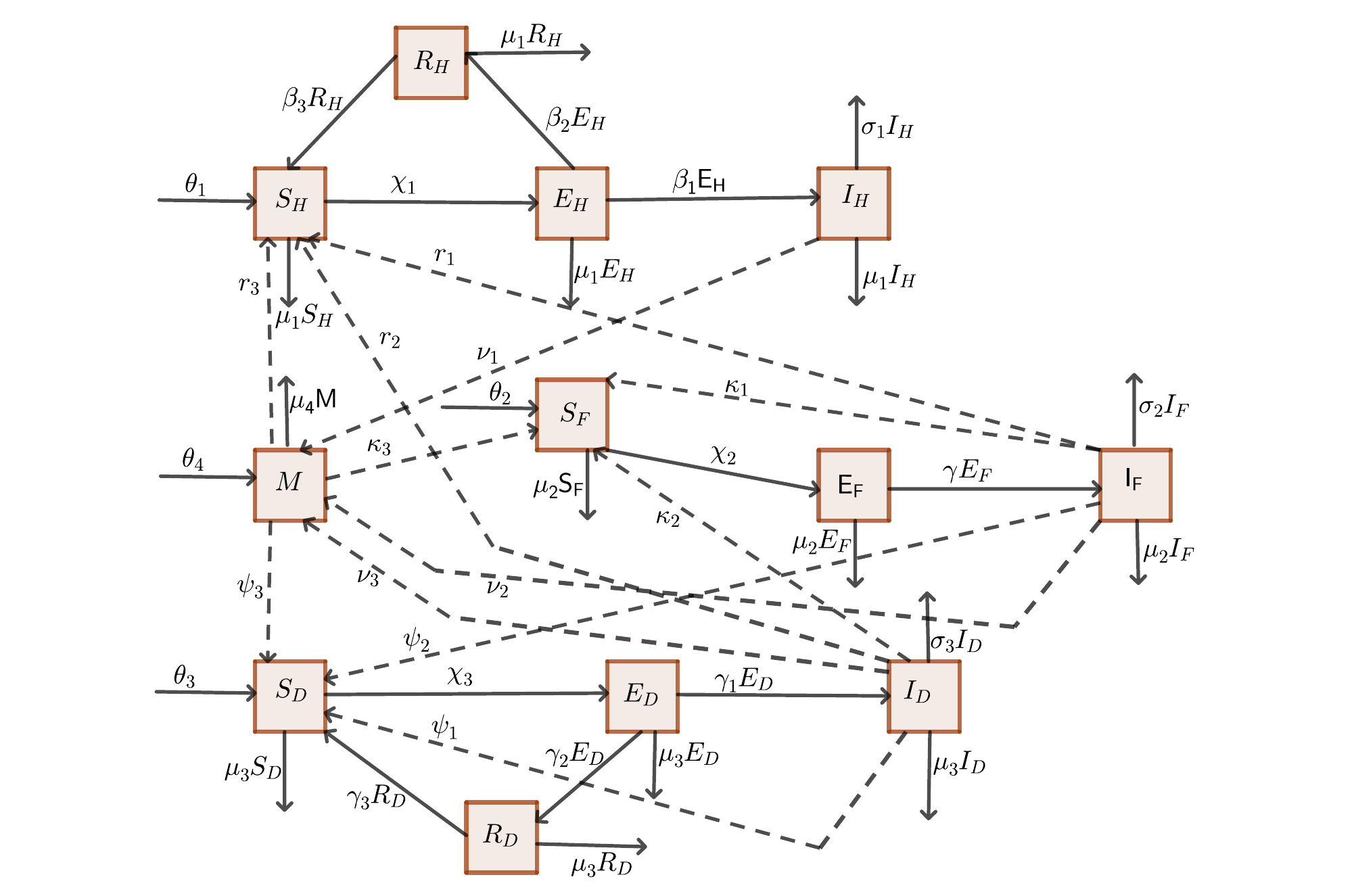}
\caption{Schematic diagram for the flow of transmission  dynamics  of rabies.}
\label{fig:Model}
\end{figure}
Using the details provided about the model parameters and their associations with the state 
variables, we develop a model represented as a set of ordinary differential equations, 
as shown in equation \eqref{eqn1},
\begin{equation}
\begin{cases}
\dot{S}_{H} &= \theta_{1} + \beta_{3} R_{H} -\chi_{1}-\mu_{1}S_{D} ,\\
\dot{E}_{H} &= \chi_{2} - (\mu_{1} + \beta_{1} + \beta_{2}) E_{H}, \\
\dot{I}_{H} &= \beta_{1}E_{H} - (\sigma_{1} + \mu_{1}) I_{H}, \\
\dot{R}_{H} &= \beta_{2} E_{H} - (\beta_{3} + \mu_{1}) R_{H}, \\
\dot{S}_{F} &= \theta_{2} - \chi_{2}-\mu_{2}S_{F}, \\
\dot{E}_{F} &= \chi_{2} - (\mu_{2} + \gamma) E_{F}, \\
\dot{I}_{F} &= \gamma E_{F} - (\mu_{2} + \sigma_{2}) I_{F} \\
\dot{S}_{D} &= \theta_{3} - \chi_{3}- \mu_{3}S_{D} + \gamma_{3} R_{D} ,\\
\dot{E}_{D} &= \chi_{3}- (\mu_{3} + \gamma_{1} + \gamma_{2}) E_{D}, \\
\dot{I}_{D} &= \gamma_{1}E_{D} - (\mu_{3} + \sigma_{3}) I_{D} ,\\
\dot{R}_{D} &= \gamma_{2}E_{D} - (\mu_{3} + \gamma_{3}) R_{D}, \\
\dot{M} &= \left(\nu_{1} I_{H} + \nu_{2} I_{F} + \nu_{3} I_{D}\right) - \mu_{4} M,
\end{cases}
\label{eqn1}
\end{equation}
subject to the following  non-negative  conditions:  
\begin{equation}
\begin{aligned}
&S_{H}(0) > 0,\; E_{H}(0) \geq 0, \; I_{H}(0) \geq 0, \; R_{H}(0) \geq 0,\; 
&S_{F}(0) > 0, \; E_{F}(0) \geq 0, \; I_{F}(0) \geq 0, \\
&S_{D}(0) \geq 0, \; E_{D}(0) \geq 0, \; I_{D}(0) \geq 0, \; R_{D}(0) \geq 0. \\
\end{aligned}
\end{equation}
The schematic diagram corresponding to \eqref{eqn1} is given in Fig.~\ref{fig:Model}.

% ---------------------------------------------

\section{Qualitative analysis}
\label{sec:02}

In this section we prove existence of a positive solution, its boundedness, 
we compute the Rabies free equilibrium and the basic reproduction number.

% ----------------------

\subsection{Positivity of the solution}	

For the model system \eqref{eqn1} to be epidemiologically meaningfully and well-posed, 
we need to prove that the state variables are non-negative  $\forall t\ge0$.

\begin{lemma}
All solution of the system in the region  \eqref{eqn1} that  start  
in $\Omega \subset{\mathbb R }^{12}_{+}$ remain  positive all the time. 
\label{theorem:1}
\end{lemma}

\begin{proof}
To prove for the existence of model solution \eqref{eqn1}, we consider initial  
conditions  and apply the integral operator $\int\limits_{0}^{t} \left(\cdot\right)ds$ 
to all equations in the model equation \eqref{eqn1}, to obtain:
\begin{equation}
\left\{
\begin{array}{llll}
S_{H}\left(t\right)-S_{H}\left(0\right)&=\int\limits_{0}^{t} \left(\theta_{1}+\beta_{3}R_{H}-\mu_{1} S_{H}-\chi_{1}\right)ds,\\
E_{H}\left(t\right)-E_{H}\left(0\right)&=\int\limits_{0}^{t}\left(\chi_{1}-\left(\mu_{1}+\beta_{1}+\beta_{2}\right)E_{H}\right)ds,\\
I_{H}\left(t\right)-I_{H}\left(0\right)&=\int\limits_{0}^{t}\left(\beta_{1}E_{H}-\left(\sigma_{1}+\mu_{1}\right) I_{H}\right)ds,\\
R_{H}\left(t\right)-R_{H}\left(0\right)&=\int\limits_{0}^{t}\left(\beta_{2} E_{H}-\left(\beta_{3}+\mu_{1} \right) R_{H}\right)ds,\\\\
S_{F}\left(t\right)-S_{F}\left(0\right)&=\int\limits_{0}^{t}\left(\theta_{2}-\chi_{2}-\mu_{2}S_{F}\right)ds,\\
E_{F}\left(t\right)-E_{F}\left(0\right)&=\int\limits_{0}^{t}\left(\chi_{2}-\left(\mu_{2}+\gamma\right)E_{F}\right)ds,\\
I_{F}\left(t\right)-I_{F}\left(0\right)&=\int\limits_{0}^{t}\left(\gamma E_{F}-\left(\mu_{2}+\sigma_{2}\right)I_{F}\right)ds,\\\\
S_{D}\left(t\right)-S_{D}\left(0\right)&=\int\limits_{0}^{t}\left(\theta_{3}-\mu_{3}S_{D}-\chi_{3} +\gamma_{3}R_{D}\right)ds,\\
E_{D}\left(t\right)-E_{D}\left(0\right)&=\int\limits_{0}^{t}\left(\chi_{3} -\left(\mu_{3}+\gamma_{1}+\gamma_{2}\right) E_{D}\right)ds,\\
I_{D}\left(t\right)-I_{D}\left(0\right)&=\int\limits_{0}^{t}\left(\gamma_{1}E_{D}-\left(\mu_{3}+\sigma_{3}\right) I_{D}\right)ds,\\
R_{D}\left(t\right)-R_{D}\left(0\right)&=\int\limits_{0}^{t}\left(\gamma_{2}E_{D}-\left(\mu_{3}+\gamma_{3}\right)R_{D}\right)ds,\\\\
M\left(t\right)-M\left(0\right)&=\int\limits_{0}^{t}\left(\left(\nu_1I_H+\nu_2I_F+\nu_3I_D\right)-\mu_4M\right)ds.
\end{array}
\right.
\label{eqn84}
\end{equation}

For convenience, from equation \eqref{eqn84} we define the  following functions:
\begin{equation}
\left\{
\begin{array}{llll}
f_1\left(t,\;S_H\right)= \theta_{1}+\beta_{3}R_{H}-\mu_{1} S_{H}-\chi_{1},\\
f_2\left(t,\;E_H\right)=\chi_{1}-\left(\mu_{1}+\beta_{1}+\beta_{2}+u_{4}\right)E_{H},\\
f_3\left(t,\;E_H\right)=\beta_{1}E_{H}-\left(\sigma_{1}+\mu_{1}\right) I_{H},\\
f_4\left(t,\;R_H\right)=\left(\beta_{2}+u_4\right) E_{H}-\left(\beta_{3}+\mu_{1} \right) R_{H},\\\\
f_5\left(t,\;S_F\right)=\theta_{2}-\chi_{2}-\mu_{2}S_{F},\\
f_6\left(t,\;E_F\right)=\chi_{2}-\left(\mu_{2}+\gamma\right)E_{F},\\
f_7\left(t,\;I_F\right)=\gamma E_{F}-\left(\mu_{2}+\sigma_{2}\right)I_{F},\\\\
f_8\left(t,\;S_D\right)=\theta_{3}-\mu_{3}S_{D}-\chi_{3} +\gamma_{3}R_{D},\\
f_9\left(t,\;E_D\right)=\chi_{3} -\left(\mu_{3}+\gamma_{1}+\gamma_{2}\right) E_{D},\\
f_{10}\left(t,\;I_D\right)=\gamma_{1}E_{D}-\left(\mu_{3}+\sigma_{3}\right) I_{D},\\
f_{11}\left(t,\;R_D\right)=\gamma_{2}E_{D}-\left(\mu_{3}+\gamma_{3}\right)R_{D},\\\\
f_{12}\left(t,\;M\right)=\left(\nu_1I_H+\nu_2I_F+\nu_3I_D\right)-\mu_4M.
 \end{array}
\right.
\end{equation} 
Since $S_H$, $E_H$, $I_H$, $R_H$, $S_F$, $E_F$, $I_F$, $S_D$, $E_D$,
$I_D$, $R_D$, $M$ are  positive and bounded in the region 
$\Omega\subset \mathbb{R}^{12}_{+}$, there exits non-negative values $\mathbb{L}_{i},\;i=1,\;2,\;\ldots, 12$, such  that
\begin{multline}
\begin{aligned}
\left\Vert f_1(t,\;S_H\left(t\right)) - f_1(t,\;S_H\left(t\right)_{1}) \right\Vert \leq 
\left\Vert \left(\theta_{1}+\beta_{3}R_{H}\left(t\right)-\mu_{1} S_{H}\left(t\right)-\chi_{1}\left(t\right)\right) \right. \\
- \left. \left(\theta_{1}+\beta_{3}R_{H}\left(t\right)-\mu_{1} S_{H}\left(t\right)_{1}-\chi_{1}\left(t\right)\right)_{1} \right\Vert\\
\leq 
\mu_{1}\left\Vert  S_{H}\left(t\right)_{1} 
-  S_{H}\left(t\right) \right\Vert\ +\chi_{1}\left\Vert  S_{H}\left(t\right)_{1} 
-  S_{H}\left(t\right) \right\Vert\ \\
\leq 
\left(\mu_{1}+\chi_{1}\right)\left\Vert  S_{H}\left(t\right)_{1} 
-  S_{H}\left(t\right) \right\Vert\ \\
\leq
\epsilon_{1}\left\Vert  S_{H}\left(t\right)_{1} 
-  S_{H}\left(t\right) \right\Vert.
\end{aligned}
\label{eqn85}
\end{multline}
Using similar approach as used in equation \eqref{eqn85}  we obtain for functions $f_i,\;i=2,3,\ldots,12$, that
\begin{multline}
\left\{
\begin{aligned}
\left\Vert f_2(t,\;E_H\left(t\right)) - f_2(t,\;E_H\left(t\right)_{1}) \right\Vert \leq\epsilon_{2}\left\Vert  E_{H}\left(t\right)_{1} 
-  E_{H}\left(t\right) \right\Vert\ ,\\
\left\Vert f_3(t,\;I_H\left(t\right)) - f_3(t,\;I_H\left(t\right)_{1}) \right\Vert \leq\epsilon_{3}\left\Vert  I_{H}\left(t\right)_{1} 
-  I_{H}\left(t\right) \right\Vert\ , \\
\left\Vert f_4(t,\;R_H\left(t\right)) - f_4(t,\;R_H\left(t\right)_{1}) \right\Vert \leq\epsilon_{4}\left\Vert  R_{H}\left(t\right)_{1} 
-  R_{H}\left(t\right) \right\Vert\ ,\\
\left\Vert f_5(t,\;S_F\left(t\right)) - f_5(t,\;S_F\left(t\right)_{1}) \right\Vert \leq\epsilon_{5}\left\Vert  S_{F}\left(t\right)_{1} 
-  S_{F}\left(t\right) \right\Vert\ , \\
\left\Vert f_6(t,\;E_F\left(t\right)) - f_6(t,\;E_F\left(t\right)_{1}) \right\Vert \leq\epsilon_{6}\left\Vert  E_{F}\left(t\right)_{1} 
-  E_{F}\left(t\right) \right\Vert\ , \\
\left\Vert f_7(t,\;I_F\left(t\right)) - f_7(t,\;I_F\left(t\right)_{1}) \right\Vert \leq\epsilon_{7}\left\Vert  I_{F}\left(t\right)_{1} 
-  I_{F}\left(t\right) \right\Vert\ , \\
\left\Vert f_8(t,\;S_D\left(t\right)) - f_8(t,\;S_D\left(t\right)_{1}) \right\Vert \leq\epsilon_{8}\left\Vert  S_{D}\left(t\right)_{1} 
-  S_{D}\left(t\right) \right\Vert\ , \\
\left\Vert f_9(t,\;E_D\left(t\right)) - f_9(t,\;E_D\left(t\right)_{1}) \right\Vert \leq\epsilon_{9}\left\Vert  E_{D}\left(t\right)_{1} 
-  E_{D}\left(t\right) \right\Vert\ , \\
\left\Vert f_{10}(t,\;I_D\left(t\right)) - f_{10}(t,\;I_D\left(t\right)_{1}) \right\Vert \leq\epsilon_{10}\left\Vert  I_{D}\left(t\right)_{1} 
-  I_{D}\left(t\right) \right\Vert\ , \\
\left\Vert f_{11}(t,\;S_D\left(t\right)) - f_{11}(t,\;R_D\left(t\right)_{1}) \right\Vert \leq\epsilon_{11}\left\Vert  R_{D}\left(t\right)_{1} 
-  R_{D}\left(t\right) \right\Vert\ , \\
\left\Vert f_{12}(t,\;M\left(t\right)) - f_{12}(t,\;M\left(t\right)_{1}) \right\Vert \leq\epsilon_{12}\left\Vert  M\left(t\right)_{1} 
-  M\left(t\right) \right\Vert,
\end{aligned}
\right.
\end{multline}
where  $\epsilon_{1}=\mu_{1}+\chi_{1}$, $\epsilon_{2}=\mu_{1}+\beta_{1}+\beta_{2}+u_{4}$, 
$\epsilon_{3}=\sigma_{1}+\mu_{1}$, $\epsilon_{4}=\beta_{3}+\mu_{1}$, 
$\epsilon_{5}=\mu_{2}+\chi_{2}$, $\epsilon_{6}=\mu_{2}+\gamma$, $\epsilon_{7}=\mu_{2}+\sigma_{1}$, 
$\epsilon_{8}=\mu_{3}+\chi_{3}$, $\epsilon_{9}=\mu_{3}+\gamma_{1}+\gamma_2+u_4$, 
$\epsilon_{10}=\mu_{3}+\sigma_{3}$, $\epsilon_{11}=\mu_{3}+\gamma_{3}$, $\epsilon_{12}=\mu_{4}$. 
Since the Lipschitz condition is satisfied, then a solution to our model \eqref{eqn1} exists.
\end{proof}

% ----------------------

\subsection{Boundedness of the model solution}

\begin{theorem}
\label{The1}	
The solution of the rabies model system \eqref{eqn1} is uniformly bounded if 
$\Omega \in \mathbb{R}_{+}^{12}$ \textup{and} $\Omega = \Omega_{H} \cup \Omega_{D} 
\cup \Omega_{F} \cup \Omega_{M} \in \mathbb{R}_{+}^{4} \times \mathbb{R}_{+}^{3} 
\times \mathbb{R}_{+}^{4} \times \mathbb{R}_{+}^{1}$, where
\begin{equation*}
\begin{split}
\Omega_{H} &= \left\{\left(S_{H}, E_{H}, I_{H}, R_{H}\right)\in\mathbb{R}_{+}^{4}: 0\leq N_{H}\leq\frac{\theta_{1}}{\mu_{1}}\right\}, \;
\Omega_{F} = \left\{\left(S_{F}, E_{F}, I_{F} \right)\in\mathbb{R}_{+}^{3}: 0\leq N_{F}\leq\frac{\theta_{2}}{\mu_{2}}\right\}, \\
\Omega_{D} &= \left\{\left(S_{D}, E_{D}, I_{D}, R_{D}\right)\in\mathbb{R}_{+}^{4}: 0\leq N_{D}\leq\frac{\theta_{3}}{\mu_{3}}\right\}, \;
\Omega_{M}  = \max\left\{\frac{\theta_1 \nu_1}{\mu_1\mu_4}+\frac{\theta_2 \nu_2}{\mu_2\mu_4}+\frac{\theta_3 \nu_3}{\mu_3\mu_4},M\left(0\right)\right\},
\end{split}
\end{equation*}
and $\Omega$ is the positive invariant region.
\end{theorem}

\begin{proof}
Since the model \eqref{eqn1} monitors both human and dog populations, we assume that the model's 
state variables and parameters are non-negative for all $t \geq 0$. 
By utilizing Theorem~\ref{The1}, we derive the invariant region of the rabies model as follows.
Consider the population  of the human  from  equation \eqref{eqn1} as
\begin{equation}
\frac{dN_{H}}{dt}=\frac{dS_{H}}{dt}+\frac{dE_{H}}{dt}+\frac{dI_{H}}{dt}+\frac{dR_{H}}{dt}.
\label{eqn7}
\end{equation}
Then,  equation \eqref{eqn7} implies to  equation \eqref{eqn8} as follows:
\begin{eqnarray}
\frac{dN_{H}}{dt}&= \theta_{1}-\left(S_{H}+E_{H}+I_{H}+R_{H}\right)\mu_{1}-\sigma_{1}I_{H}.
\label{eqn8}
\end{eqnarray}  
But if $N_{H}=S_{H}+E_{H}+I_{H}+R_{H}$, then Equation \eqref{eqn8} can be written as
\begin{equation}
    \frac{dN_{H}}{dt} = \theta_{1} - N_{H}\mu_{1}, \label{eqn11}
\end{equation}
and by utilizing the integrating factor in  equation \eqref{eqn11} and $t \rightarrow 0$, we have
\begin{equation}
    N_{H}(0) \le \frac{\theta_{1}}{\mu_{1}} + Ce^{0} \implies N_{H}(0) - \frac{\theta_{1}}{\mu_{1}} \le C. \label{eqn16}
\end{equation}
By simplifying  equation \eqref{eqn16}  and after simple manipulations, it follows that
\begin{equation*}
\Omega_{H} = \left\{ \left( S_{H}, E_{H}, I_{H}, R_{H} \right) 
\in \mathbb{R}_{+}^{4} : 0 \leq N_{H} \leq \frac{\theta_{1}}{\mu_{1}} \right\}.
\end{equation*}
Employing  the same procedures for free range  and domestic dog populations, we obtain:  
\begin{equation}
\begin{aligned}
\Omega_{F} &= \left\{ (S_{F}, E_{F}, I_{F}) \in \mathbb{R}_{+}^{3} : 0 \leq N_{F} \leq \frac{\theta_{2}}{\mu_{2}} \right\}, \;
\Omega_{D} = \left\{ (S_{D}, E_{D}, I_{D}, R_{D}) \in \mathbb{R}_{+}^{4} : 0 \leq N_{D} \leq \frac{\theta_{3}}{\mu_{3}} \right\}.
\end{aligned}
\end{equation}
Again, from the environment that contains  rabies virus,
\begin{equation}
\begin{aligned}
\dot{M} = \left(\nu_1 I_H + \nu_2 I_F + \nu_3 I_D\right) - \mu_4 M. 
\label{eqn25}
\end{aligned}
\end{equation}
Since $N_H \leq \frac{\theta_1}{\mu_1}$, $N_F \leq \frac{\theta_2}{\mu_3}$, 
and $N_D \leq \frac{\theta_3}{\mu_3}$,
it follows that $I_H \leq \frac{\theta_1}{\mu_1}$, $I_F \leq \frac{\theta_2}{\mu_2}$, 
and $I_D \leq \frac{\theta_3}{\mu_3}$. 
Therefore, Equation \eqref{eqn25} can be expressed as
\begin{equation}
\dot{M} \leq \left(\frac{\nu_1 \theta_1}{\mu_1} + \frac{\nu_2 \theta_2}{\mu_2} 
+ \frac{\nu_3 \theta_3}{\mu_3}\right) - \mu_4 M. 
\label{eqn26}
\end{equation}
Now, let $Y$ be the solution which is unique to the initial value problem  such that 
\begin{equation}
\left.
\begin{array}{llll}
\dot{Y}\leq\left(\frac{\nu_1\theta_1}{\mu_1}+\frac{\nu_2\theta_2}{\mu_2}
+\frac{\nu_3\theta_3}{\mu_3}\right)-\mu_4M,\;\; \textup{for} \;\;\; t>0\\
Y(0)=M(0).
\end{array}
\right\}
\label{eqn28}
\end{equation}
By using  integration factor  equation \eqref{eqn28} as $t\rightarrow \infty$, we get the expression  
\begin{equation}
\left( M(0) - \left( \frac{\nu_1\theta_1}{\mu_1} + \frac{\nu_2\theta_2}{\mu_2} 
+ \frac{\nu_3\theta_3}{\mu_3} \right) \frac{1}{\mu_{4}} e^{\mu_{4} t} \right)
\label{eqn30}
\end{equation}
and as  equation \eqref{eqn30} goes to zero, we have              
\begin{eqnarray} 
M\left(t\right)\le\Omega_{M}.
\end{eqnarray}
Thus, the model system  \eqref{eqn1} is biologically and mathematically meaningfully 
such that their solution relies in the region  $\Omega$.
\end{proof}

% ----------------------

\subsection{Rabies free equilibrium ($\mathbb{E}_0$) and the basic reproduction number ${\cal R}_0$}

In order to achieve a rabies-free equilibrium ($\mathbb{E}_0$) in both humans,  
free-range, and domestic dog populations, we equate all infectious compartments 
in equation \eqref{eqn1} to zero, which leads to
\begin{align*}
{\mathbb{E} }_{0} &=   \left(\frac{\theta_{1}}{\mu_{1}},0,0,0,
\frac{\theta_{2}}{\mu_{2}},0,0,\frac{\theta_{3}}{\mu_{3}},0,0,0,0\right).
\end{align*}
	
In order to determine the  basic reproduction   ${\cal R}_0$,  the  next generation matrix method, as applied 
by \cite{lasalle1976stability,yang2014basic,saha2021dynamics}, is adopted as follows:
\begin{eqnarray}
\dfrac{dx_{i}}{dt} &=\mathcal{F}_{i}\left(x\right)
-\left(\mathcal{V}_{i}^{+}\left(x\right)-\mathcal{V}_{i}^{-}\left(x\right)\right),
\label{eqn31}
\end{eqnarray} 
where $\mathcal{F}_{i}$ is the matrix new infections in the compartment $i$  
while  $\mathcal{V}_{i}^{+}$ and $\mathcal{V}_{i}^{-}$ 
are  matrices  of the transfer terms in and out of the compartment $i$, respectively.  
From equation \eqref{eqn31}, we define
$\mathcal{F}_{i}$ and $\mathcal{V}_{i}$ by
\begin{eqnarray}
\mathcal{F}_{i}=
\left(
\begin{array}{c}
\left(\tau_{1}I_{F}+\tau_{2}I_{D}+\tau_{3} \lambda \left(M\right)\right)S_{H}   \\
0\\
\left(\kappa_{1}I_{F}+\kappa_{2} I_{D}+\kappa_{3} \lambda \left(M\right)\right)S_{F}\\
0\\
\left(\frac{\psi_{1}I_{F}}{1+\rho_{1}}+\frac{\psi_{2}I_{D}}{1+\rho_{2}}
+\frac{\psi_{3}}{1+\rho_{3}}\lambda \left(M\right)\right) S_{D}\\
0\\
0
\end{array}
\right),\quad
\mathcal{V}_{i}=
\left(
\begin{array}{c}
\left(\mu_{1}+\beta_{1}+\beta_{2}\right)E_{H} \\
\left(\sigma_{1}+\mu_{1}\right)I_{H}-\beta_{1}E_{H}\\
\left(\mu_{2}+\gamma\right)E_{F}\\
\left(\mu_{2}+\sigma_{2}\right)I_{F}- \gamma E_{F}\\
\left(\mu_{3}+\gamma_{1}+\gamma_{2}\right) E_{D}\\
\left(\mu_{3}+\delta_{3}\right) I_{D}-\gamma_{1}E_{D}\\
\mu_4M-\left(\nu_1I_H+\nu_2I_F+\nu_3I_D\right)
\end{array}
\right).
\label{M}
\end{eqnarray}
The Jacobian matrices $F$ and $V$ at the disease free equilibrium point ${\mathbb E}_{0}$ 
are given by the expression of matrix  $F V^{-1}$ which  can be presented  as
\begin{equation}
\begin{aligned}
\setlength{\arraycolsep}{10pt}
F V^{-1} = \begin{pmatrix}
0 & 0 & R_{13} & R_{14} & R_{15} & R_{16} & 0 \\
0 & 0 & 0 & 0 & 0 & 0 & 0 \\
0 & 0 & R_{33} & R_{34} & R_{35} & R_{36} & 0 \\
0 & 0 & 0 & 0 & 0 & 0 & 0 \\
0 & 0 & R_{53} & R_{54} & R_{55} & R_{56} & 0 \\
0 & 0 & 0 & 0 & 0 & 0 & 0 \\
0 & 0 & 0 & 0 & 0 & 0 & 0 \\
\end{pmatrix},
\label{eqn24}
\end{aligned}
\end{equation}
where  
\begin{equation}
\left.
\begin{array}{llll}
R_{13}={\dfrac {\tau_{{1}}\theta_{{1}}
\gamma}{\mu_{{1}} \left( \mu_{{2}}+\gamma \right)  \left( \mu_{{2}}+
\sigma_{{2}} \right) }},\;\; R_{14}={\frac {\tau_{{1}}\theta_{{1}}}{\mu_{{1}}
\left( \mu_{{2}}+\sigma_{{2}} \right) }},\;\;\;
R_{15}={\dfrac {\tau_{{2}}\theta_{{1
}}\gamma}{\mu_{{1}} \left( \mu_{{3}}+\gamma_{{1}}+\gamma_{{2}}
\right)  \left( \mu_{{3}}+\sigma_{{3}} \right) }},\\       
R_{33}=\dfrac{\kappa_{1}\theta_{2}\gamma}{\mu_{2}\left(\mu_{2}+\gamma\right)
\left(\mu_{2}+\sigma_{2}\right)},\;\;\; R_{34}={\dfrac {\kappa_{{1}
}\theta_{{2}}}{\mu_{{2}} \left( \mu_{{2}}+\sigma_{{2}} \right) }},\;\;\; 
R_{35}=\dfrac{\kappa_{1}\theta_{2}\gamma}{\mu_{2}\left(\mu_{3}+\gamma_{1}
+\gamma_{2}\right)\left(\mu_{3}+\sigma_{3}\right)},\\ 
R_{53}=\dfrac{\psi_{1}\theta_{3}\gamma}{\left(1+\rho_{1}\right)\left(\mu_{2}+\gamma\right)
\left(\mu_{2}+\sigma_{2}\right)\mu_{3}} , \;\;\; R_{54}={\dfrac {\psi_{{1}}
\theta_{{3}}}{ \left( 1+\rho_{{1}} \right) \mu_{{3}} \left( \mu_{{2}}
+\sigma_{{2}} \right) }},\;\;\; R_{16}= {\dfrac {\tau_{{2}}
\theta_{{1}}}{\mu_{{1}} \left( \mu_{{3}}+\sigma_{{3}} \right) }},\\  
R_{55}=\dfrac{\psi_{2}\theta_{3}\gamma}{\left(1+\rho_{2}\right)
\left(\mu_{3}+\gamma_{1}+\gamma_{2}\right)\left(\mu_{3}+\sigma_{3}\right)\mu_{3}},\;
R_{56}={\dfrac 
{\psi_{{2}}\theta_{{3}}}{ \left( 1+\rho_{{2}} \right) \mu_{{3}}
\left( \mu_{{3}}+\sigma_{{3}} \right) }},\;\;
R_{36}={\dfrac {\kappa_{{1}}\theta_{{2}}}{\mu_{{2}} 
\left( \mu_{{3}}+\sigma_{{3}} \right) }}. 
\end{array}
\right\}
\label{eqn10}
\end{equation}
From  equation \eqref{eqn24}, we obtain the eigenvalues  as
\begin{gather}
\begin{aligned}
&\left(0, 0, 0, 0, 0, \frac{1}{2} R_{{55}} + \frac{1}{2} R_{{33}} 
+ \frac{1}{2} \sqrt{R_{{33}}^2 - 2 R_{{33}} R_{{55}} + 4 R_{{35}} R_{{53}} + R_{{55}}^2}, \right. \\
&\left. \frac{1}{2} R_{{55}} + \frac{1}{2} R_{{33}} - \frac{1}{2} 
\sqrt{R_{{33}}^2 - 2 R_{{33}} R_{{55}} + 4 R_{{35}} R_{{53}} + R_{{55}}^2} \right).
\label{eqn3}
\end{aligned}
\end{gather}
This non-negative eigenvalue corresponds to a non-negative eigenvector that represents 
the distribution of infected individuals who generate the highest number of secondary 
infections per generation, also known as ${\cal R}_0$. According to 
\cite{dharmaratne2020estimation}, the basic reproduction  number ${\cal R}_0$  
is the largest eigenvalue of the next generating matrix given by
\begin{eqnarray}
{\cal R}_0&=\rho \left(FV^{-1}\right).
\end{eqnarray}  
Therefore, the spectral radius of the next generation matrix is given by
\begin{equation}
{\cal R}_0= \rho\left(FV^{-1}\right)=\frac{\left(R_{55}+R_{33}\right)
+\sqrt{R_{33}\left(R_{33}-2R_{55}\right)+4R_{35}R_{53}+R_{55}^{2}}}{2}.
\label{R0}
\end{equation}
   
% ----------------------------------------------------------

\subsubsection{Global stability of the Rabies free equilibrium}

The model behaviour at the disease free equilibrium point, ${\mathbb E}_0$, is investigated 
using a Metzler matrix as applied by Castillo-Chavez et al.  \cite{castillo2002computation}.  
Let $X_{s}$ denote non-transmitting class, $X_{m}$ be transmitting class and $X_{DFE}$ 
be the Disease Free equilibrium, whereby
\begin{equation}
\left.
\begin{aligned}
\dfrac{dX_{s}}{dt}&=B\left(X_{s}-X_{DFE}\right)+B_{1} X_m\\
\dfrac{dX_{m}}{dt}&=B_{2}X_{m}
\end{aligned}
\right\}.
\end{equation}
Then,  from the model system \eqref{eqn1},  it can be deduced that
\begin{equation*}
\begin{aligned}
X_{s}&=\left(S_{H},\;R_{H},\;S_{F},\;S_{D},\; R_{D}\right)^{T}, 
\quad X_{m}=\left(E_{H},\;I_{H},\;E_{D},\;I_{D},M\right)^{T},\\
X_{s}-X_{DFE}
&=
\begin{bmatrix}
S_H-\dfrac{\theta_1}{\mu_{1}}\\
R_H\\
S_F-\dfrac{\theta_2}{\mu_{2}}\\
S_D-\dfrac{\theta_3}{\mu_{3}}\\
R_D
\end{bmatrix},\;\; 
B=\begin{bmatrix}
-\mu & \beta_{3}&0&0&0\cr
0&-\left(\beta_{3}+\mu_{1}\right)&0&0&0\cr
0&0&-\mu_{2}&0&0\cr
0&0&0&-\mu_{3}&\gamma_{3}\cr
0&0&0&0&-\left(\mu_{3}+\gamma_{3}\right)
\end{bmatrix},
\end{aligned}
\end{equation*}
\begin{equation*}
\begin{aligned}
B_{1}&=\begin{bmatrix}
0&0&0&{\dfrac {\tau_{{1}}\theta_{{1}}}{
\mu_{{1}}}}&0&{\dfrac {\tau_{{2}}\theta_{{1}}}{\mu_{{1}}}}&0\\ 
\noalign{\medskip}\beta_{{2}}&0&0&0&0&0&0\cr \noalign{\medskip}0&0&0
&{\frac {\kappa_{{1}}\theta_{{2}}}{\mu_{{2}}}}&0&{\dfrac {\kappa_{{2}}
\theta_{{2}}}{\mu_{{2}}}}&0\cr \noalign{\medskip}0&0&0&{\dfrac {\psi_{{1}}
\theta_{{3}}}{\mu_{{3}} \left( 1+\rho_{{1}} \right) }}&0
&{\dfrac {\psi_{{2}}\theta_{{3}}}{\mu_{{3}} \left( 1+\rho_{{2}} \right) }}&0\cr 
\noalign{\medskip}0&0&0&0&\gamma_{{2}}&0&0
\end{bmatrix}, \;\;\textup{and}\\\\ 
B_{2}&=\begin{bmatrix}
-\mu_{{1}}-\beta_{{1}}-\beta_{{2}}&0&0
&{\frac {\tau_{{1}}\theta_{{1}}}{\mu_{{1}}}}&0&{\frac {\tau_{{2}}
\theta_{{1}}}{\mu_{{1}}}}&0\\ \noalign{\medskip}\beta_{{1}}&-\sigma_{{1
}}-\mu_{{1}}&0&0&0&0&0\\ \noalign{\medskip}0&0&-\mu_{{2}}-\gamma&{
\frac {\kappa_{{1}}\theta_{{2}}}{\mu_{{2}}}}&0&{\frac {\kappa_{{1}}
\theta_{{2}}}{\mu_{{2}}}}&0\\ \noalign{\medskip}0&0&\gamma&-\mu_{{2}}-
\sigma_{{2}}&0&0&0\\ \noalign{\medskip}0&0&0&{\frac {\psi_{{1}}\theta_
{{3}}}{\mu_{{3}} \left( 1+\rho_{{1}} \right) }}&-\mu_{{3}}-\gamma_{{1}
}-\gamma_{{2}}&{\frac {\psi_{{2}}\theta_{{3}}}{\mu_{{3}} \left( 1+\rho
_{{1}} \right) }}&0\\ \noalign{\medskip}0&0&0
&0&\gamma&-\mu_{{3}}-\sigma_{{3}}&0\\ \noalign{\medskip}0&\nu_{{1}}&0&
\nu_{{2}}&0&\nu_{{3}}&-\mu_{{4}}
\end{bmatrix}.
\end{aligned}
\end{equation*}

The non-negative out-diagonal entries of matrix $B_{2}$ indicate that it is the Metzler matrix.  
To prove the stability of matrix \( B_{2} \) using the Metzler matrix block form,  
we need to demonstrate that it satisfies the conditions for a Metzler matrix 
and has eigenvalues with negative real parts.  
We adopt the idea of a stable Metzler matrix and apply the following Lemma~\ref{Lema}.

\begin{lemma}[See \cite{kamgang2008computation}]
\label{Lema}
Let $N$ be a square Metzler matrix written in block form
\begin{equation}
N=
\begin{bmatrix}
A && B \\
C && D \\
\end{bmatrix},
\end{equation}
where $A$, $B$ and  $D$  are rectangular  matrices while $C$ is a square matrix. 
Then $N$ is Metzler stable if, and only  if, matrices $A$ and  $D - CA^{-1} B$ 
are Metzler stable.
\end{lemma}

In our case one has
\begin{equation*}
A = 
\begin{bmatrix}
-\mu_{1}-\beta_{1}-\beta_{2} & 0 & 0 \\
\beta_{1} & -\sigma_{1}-\mu_{1} & 0 \\
0 & 0 & -\mu_{2}-\gamma \\
\end{bmatrix},\;\\
B = 
\begin{bmatrix}
\frac{\tau_{1}\theta_{1}}{\mu_{1}} & 0 & \frac{\tau_{2}\theta_{1}}{\mu_{1}} & 0 \\
0 & 0 & 0 & 0 \\
\frac{\kappa_{1}\theta_{2}}{\mu_{2}} & 0 & \frac{\kappa_{1}\theta_{2}}{\mu_{2}} & 0 \\
0 & 0 & 0 & 0 \\
\frac{\psi_{1}\theta_{3}}{\mu_{3}(1+\rho_{1})} & 0 & \frac{\psi_{2}\theta_{3}}{\mu_{3}(1+\rho_{1})} & 0 \\
0 & 0 & 0 & 0 \\
\nu_{1} & 0 & \nu_{2} & 0 \\
\end{bmatrix},\;\\C = 
\begin{bmatrix}
0 & 0 & 0 \\
\gamma & 0 & 0 \\
0 & 0 & 0 \\
\end{bmatrix},
\end{equation*}
and 
\begin{equation*}
D = 
\begin{bmatrix}
-\mu_{2}-\sigma_{2} & 0 & 0 \\
0 & -\mu_{3}-\gamma_{1}-\gamma_{2} & 0 \\
0 & \gamma & -\mu_{3}-\sigma_{3} \\
\end{bmatrix}.
\end{equation*}
To show that matrices \( A \) and \( D - CA^{-1}B \) are Metzler stable, 
we need to demonstrate that all their off-diagonal elements are non-negative. 
Since the off-diagonal elements of \( A \) are \( \beta_{1} \) and \( 0 \), 
both of which are non-negative, we conclude that matrix \( A \) is Metzler stable.

Next, matrix \( D - CA^{-1}B \) can be calculated as follows:
\begin{equation*}
CA^{-1}B = 
\begin{bmatrix}
\frac{\tau_{1}\theta_{1}}{\mu_{1}} & 0 & \frac{\tau_{2}\theta_{1}}{\mu_{1}} \\
0 & 0 & 0 \\
0 & \gamma & -\mu_{2} - \sigma_{2} \\
\end{bmatrix}.
\end{equation*}
Subtracting \( CA^{-1}B \) from \( D \) gives
\begin{equation*}
D - CA^{-1}B = 
\begin{bmatrix}
0 & 0 & 0 \\
0 & -\mu_{1} - \sigma_{1} & 0 \\
0 & 0 & -\mu_{2} - \gamma - \sigma_{2} \\
\end{bmatrix}.
\end{equation*}
Since the eigenvalues of the Matrix $B$ are $\lambda_{1}=-\mu_{3}$,  $\lambda_{2}=-\mu_{2}$,  
$\lambda_{3}=-\mu_{1}$, $\lambda_{4}=-\left(\mu_{3}+\gamma_{3}\right)$, and 
$\lambda_{5}=-\left(\beta_{3}+\mu_{1}\right)$ and  the off-diagonal elements 
of \( D - CA^{-1}B \) are all \( 0 \), which are non-negative, thus matrix 
\( D - CA^{-1}B \) is also Metzler stable. Since both matrices \( A \) and \( D - CA^{-1}B \) 
have non-negative off-diagonal elements, they are Metzler stable. Thus, \( B_{2} \) is also 
Metzler stable. Therefore, the rabies free equilibrium point of the model system \eqref{eqn1} 
is globally asymptotically stable  if ${\cal R}_0<1 \textup{ and unstable otherwise}$.

% --------------------------------

\subsubsection{Rabies Persistence Equilibrium Point (RPEP)}

The endemic equilibrium point is the steady state where rabies is present in humans, free-range dogs, 
and domestic dogs. To find this point, we set the equations of the model system \eqref{eqn1} 
to zero and solve the resulting system simultaneously. The state variables for each compartment are represented by
\begin{align*}
\textup{RPEP}\left(S_{H}^{*},\; E_{H}^{*},\; I_{H}^{*},\; R_{H}^{*},\; S_{F}^{*},\; E_{F}^{*},\; 
I_{F}^{*},\; S_{D}^{*},\; E_{D}^{*},\; I_{D}^{*},\; R_{D}^{*},\; M^{*}\right)\;\;  
\end{align*}  
such  that
\begin{equation*}
\left.
\begin{aligned}
R^{*}_H &= {\dfrac {\beta_{{2}} \left( \sigma_{{1}}+\mu_{{1}} \right) I^{*}_{{
H}}}{\beta_{{1}} \left( \beta_{{3}}+\mu_{{1}} \right) }}, \\\\
I^{*}_{H} &= \dfrac{\beta_1(\beta_3 + \mu_3)(\sigma_1 + \mu_1)^2(\beta_1 + \beta_2 + \beta_3)\mu_1 
+ \beta_1\beta_3(\sigma_1 + \mu_1)^2}{(\sigma_1 + \mu_1)^2((\beta_1 + \beta_2 + \beta_3)\mu_1 + \beta_1\beta_3)} \\
&\quad - \dfrac{\beta_1(\beta_3 + \mu_3)(\sigma_1 + \mu_1)^2\beta_3 - \theta_1(\beta_3 + \mu_3)
(\sigma_1 + \mu_1)^2}{(\sigma_1 + \mu_1)^2((\beta_1 + \beta_2 + \beta_3)\mu_1 + \beta_1\beta_3)},\\\\
E^{*}_H &= {\dfrac { \left( \sigma_{{1}}+\mu_{{1}} \right) I^{*}_{{H}}}{\beta_{{1}}}},\\
S^{*}_{{H}} &={\dfrac {\beta_{{3}}\beta_{{2}} \left( \sigma_{{1}}+\mu_{{1}} \right) I^{*}
_{{H}}}{\beta_{{1}} \left( \beta_{{3}}+\mu_{{1}} \right) \mu_{{1}}}}
-{\dfrac { \left( \mu_{{1}}+\beta_{{1}}+\beta_{{2}} \right)  
\left( \sigma_{{1}}+\mu_{{1}} \right) I^{*}_{{H}}}{\beta_{{1}}\mu_{{1}}}}
+{\frac{\theta_{{1}}}{\mu_{{1}}}}
\end{aligned}
\right\}.
\end{equation*}
Re-writing  Eq.~\eqref{eqn3} as
\begin{equation}
\left(R_{0} - 1\right)\left(R_{0} - R_{33} - R_{55} + 1\right) 
+ \left(1 - R_{33}\right)\left(1 - R_{55}\right) - R_{35}R_{53} = 0,
\end{equation}
where  
$R_{33}$, $R_{55}$, $R_{35}$, and $R_{53}$ are defined in Eq.~\eqref{eqn10}, we get
\begin{equation*}
\left.
\begin{aligned}
I^{*}_{D} = \dfrac{\gamma_{1}\psi_{1}I^{*}_{F}(1+\rho_{2})(1+\rho_{3})M^{*} 
+ \gamma_{1}\psi_{3}M^{*}(1+\rho_{1})(1+\rho_{2})}{(\mu_{3} + \gamma_{1} + \gamma_{2})^2 
- \gamma_{1}\psi_{2}(1+\rho_{1})(1+\rho_{3})M^{*}(\mu_{3} + \gamma_{1} + \gamma_{2})},\\\\
E^{*}_{{D}} ={\dfrac { \left( \mu_{{3}}+\sigma_{{3}} \right) I^{*}_{{D}}}{\gamma_{{1}}}},\;\;
R^{*}_{{D}}={\dfrac {\gamma_{{2}} \left( \mu_{{3}}+\sigma_{{3}} \right) 
I^{*}_{{D}}}{\gamma_{{1}} \left( \mu_{{3}}+\gamma_{{3}} \right) }},
\end{aligned}
\right\}
\end{equation*}
\begin{equation*}
\left.
\begin{aligned}
S^{*}_{{D}}&=\dfrac{{ \gamma_3 (\mu_3 + \sigma_3)I^{*}_D}}{{\mu_3 \gamma_1}}
-\dfrac{{(\mu_3 + \gamma_1 + \gamma_2) \gamma_2 (\mu_3 + \sigma_3) 
I^{*}_D}}{{\gamma_1 (\mu_3 + \gamma_3) \mu_3}}+\dfrac{\theta_{3}}{\mu_{3}},\;\;
E^{*}_{{F}}={\dfrac { \left( \mu_{{2}}+\sigma_{{2}} \right) I^{*}_{{F}}}{\gamma}},\\
S^{*}_{F}&={\dfrac {\theta_{{2}}}{\mu_{{2}}}}-{\dfrac { \left( \mu_{{2}}
+\gamma \right)  \left( \mu_{{2}}+\sigma_{{2}} \right) I^{*}_{{F}}}{\gamma\,\mu}},\;\;
M^{*}={\dfrac {\nu_{{3}}I^{*}_{{D}}+\nu_{{2}}I^{*}_{{F}}+\nu_{{1}}I^{*}_{{H}}}{\mu_{{4}}}},
\end{aligned}
\right\}
\end{equation*}
where
\begin{multline*}
\theta_2 = \dfrac{(\mu_2 + \gamma) \mu_2 \left(1 + (R_0 - 1)\right) (1 + \rho_1) 
\mu_3 (\mu_2 + \sigma_2)}{\left[ \mu_3 (1 + \rho_2) (1 + \rho_1) (\mu_3 + \sigma_3) 
(\mu_3 + \gamma_1 + \gamma_2) \left(1 + (R_0 - 1)\right) - \theta_3 \gamma_1 
\left(\psi_2 (1 + \rho_1) \mu_3 + \psi_1 (1 + \rho_2)\right) \right] \gamma \kappa_1} \\
\times \left[ (1 + \rho_2)(\mu_3 + \sigma_3) (\mu_3 + \gamma_1 + \gamma_2) 
\left(1 + (R_0 - 1)\right) - \theta_3 \psi_2 \gamma_1 \right], \\[10pt]
\theta_3 = \dfrac{\left[ -\mu_2 (\mu_2 + \sigma_2) (\mu_2 + \gamma) \left(1 + (R_0 - 1)\right) 
+ \gamma \kappa_1 \theta_2 \right] \left(1 + (R_0 - 1)\right) (1 + \rho_1) 
\mu_3 (1 + \rho_2) (\mu_3 + \sigma_3) (\mu_3 + \gamma_1 + \gamma_2)}{\left[ 
\left( -\mu_2 (\mu_2 + \sigma_2) (\mu_2 + \gamma) \left(1 + (R_0 - 1)\right) 
+ \gamma \kappa_1 \theta_2 \right) (1 + \rho_1) \psi_2 \mu_3 
+ \gamma \kappa_1 \theta_2 \psi_1 (1 + \rho_2) \right] \gamma_1}.
\end{multline*}

The endemic equilibrium point of the rabies disease persists when 
$I_{H},\; I_{F},\; I_{D},\; M > 0$ and ${\cal R}_0 \geq 1$, as summarized in Theorem~\ref{The}.

\begin{theorem}
The system \eqref{eqn1} has a unique endemic equilibrium RPEP 
if $\mathcal{R}_0 \geq 1$ and $I_{H},\; I_{F},\; I_{D},\; M > 0$.
\label{The}
\end{theorem}

% ---------------------------------------------  

\section{Quantitative analysis}
\label{sec:03}

In this study, we create artificial data sets that represents the dynamics of infectious 
diseases and estimate model parameters using the least squares technique. Baseline parameter 
values $\Theta_i$ of the rabies  model  from the literature are used to numerically solve 
a non-linear deterministic model \eqref{eqn1} in the Matlab environment, generating synthetic 
datasets at each time $t_i$. Initial conditions for the number of susceptible, exposed  
infected, and recovered humans, free range  and  Domestic dogs are also provided. 

% ----------------------

\subsection{Parameter estimation and model fitting}

In order to generate the  rabies    model-predicted number of infected individuals  
$RD(t_i; \Theta_i)$ dataset  at time $t_i$ for a given parameter vector $\Theta_i$, 
we added random Gaussian noise  $\eta_i\left(t_i\;\;\Theta_i\right)$ measurements 
to the data simulate real-world dynamics, where measurement errors are common.  
Thus, the  observed  dependent  data  were given  as
\begin{equation}
\begin{array}{llll}
I(t_i; \Theta_i)=RD\left(t_i\;\;\Theta_i
\right)+\eta_i\left(t_i\;\;\Theta_i\right) \;\; \text{for each time}\;\; t_i\in[1,\;\;n],
\end{array}
\end{equation}
where $n$ is the number of data points.
The objective function quantifies the difference between the model's predictions 
$I(t_i; \Theta_i)$ and the estimated  data $ I_{\text{estimated}, i}$. We define 
the objective function $J(\Theta_i)$  as the sum of squared differences 
between model predictions and estimated dataset as
\begin{equation}
\begin{array}{llll}
J(\Theta_i) = \displaystyle\sum\limits_{i=1}^{33} a_i\ 
\left(I(t_i; \Theta_i) - I_{\text{estimated}, i}\right)^2,  
\end{array}
\label{eqn2}
\end{equation}
where $I_{\text{estimated}, i}$ is the estimated number 
of infected  rabies individuals at time $t_i$. The term 
$$
\left(I(t_i; \Theta_i) - I_{\text{estimated}, i}\right)^2
$$ 
represents the squared difference between the model-predicted value $I(t_i; \Theta_i)$ 
and the estimated  value $I_{\text{estimated}, i}$ at time $t_i$; 
$a_i$ typically represents the weighting or importance assigned 
to each term in the summation. The goal is to find the parameter vector 
$\hat{\Theta}_i$ that minimizes the objective function $J(\Theta_i)$ given by
\begin{equation}
\begin{array}{llll}
\hat{\Theta}_i = \arg\min J(\Theta_i).  
\end{array}
\label{eqn4}
\end{equation}
The expression indicates that $\hat{\Theta}_i$ is the parameter vector that optimally 
fits the model to the estimated data by minimizing the objective function $J(\Theta_i)$. 
The symbol $\arg\min$ is used to find the argument (in this case, $\Theta_i$) that minimizes 
the function $J(\Theta_i)$.  To find the parameter value $\Theta_i$, we apply  partial  
derivatives on Eq.~\eqref{eqn1} with respect  to each of the 33 parameters  
represented  by  $\Theta_i$  such that
\begin{equation}
\begin{array}{llll}
\dfrac{\partial J(\Theta_i)}{\partial \Theta_i }
= \displaystyle\sum\limits_{i=1}^{33} 2\ \left(I(t_i; \Theta_i) - I_{\text{estimated}, i}\right)  
\dfrac{\partial (I(t_i; \Theta_i))}{\partial \Theta_i }=0\;\;\text{for} \;\;  i=1,2,\ldots,33,
\end{array}
\label{eqn6}
\end{equation}
where (see Table~\ref{T2})
\begin{equation*}
\small
\begin{aligned}
\Theta_i = & \rho_{1},\;\;\rho_{2},\;\;\rho_{3},\;\;\theta_{1},\;\;\tau_{1},\;\;
\tau_{2},\;\;\tau_{3},\;\;\beta_{1},\;\;\beta_{2},\;\;\beta_{3}, \; 
\mu_{1},\;\;\sigma_{1},\;\;\theta_{2},\;\;\gamma,\;\;\kappa_{1},\;\;\kappa_{2},\;\;
\kappa_{3},\;\;\sigma_{2},\;\;\mu_{2},\;\;\theta_{3}, \\
& \psi_{1},\;\;\psi_{2},\;\;\psi_{3},\;\;\sigma_{3},\;\;\gamma_{1},\;\;\gamma_{2},\;\;
\gamma_{3},\;\;\mu_{3},\;\;\nu_{1},\;\;\nu_{2}, \; \nu_{3},\;\;\mu_{4},\;\;C.
\end{aligned}
\end{equation*}
\begin{center}
\begin{longtable}{|l|l|l|l|l|}
\caption{\centering Estimated model parameters (Year$^{-1}$), initial guess 
for parameters (Year$^{-1}$) and their respective source.}\\ \hline 
\textbf{Parameters} & \textbf{Baseline value} & \textbf{Source} 
& \textbf{Estimated value}& $\textup{Mean}\left(\mu \right) \textup{and std}\left(\sigma\right) $\\
\hline\hline
\endfirsthead
\multicolumn{4}{c}
{\tablename\ \thetable\ -- \textit{Continued from previous page}} \\
\hline
\textbf{Parameters} & \textbf{Baseline value} & \textbf{Source} 
& \textbf{Estimated value}& $\textup{Mean}\left(\mu \right) \textup{and std}\left(\sigma\right) $\\ \hline
\endhead
\hline \multicolumn{4}{r}{\textit{Continued on next page}} \\
\endfoot
\hline
\endlastfoot

$\theta_{1}$ & 2000& (Estimated)&1993.382113 & $\mathcal{ N}\left(1996.691056\;\; 4.4679553 \right)$. \\
$\tau_{1}$ & 0.0004& \cite{tian2018transmission}&0.000405 &$\mathcal{ N}\left(0.000402\;\; 4\times10^{-6} \right)$.\\
$\tau_{2}$ & 0.0004& \cite{tian2018transmission}&0.000604&$\mathcal{ N}\left(0.000502\;\; 1.44\times10^{-4} \right)$.\\
$\tau_{3}$ & $\left[0.0003\;\;  0.0100\right]$ & (Estimated)&0.000303 &$\mathcal{ N}\left(0.000302\;\; 2\times10^{-6} \right)$ .\\
$\beta_{1}$ & $\frac{1}{6}$ & \cite{zhang2011analysis,tian2018transmission}&0.165581& $\mathcal{ N}\left(0.166124\;\; 7.68\times10^{-4} \right)$.\\
$\beta_{2}$ & $\left[0.54 \;\; 1\right]$ & \cite{zhang2011analysis,abdulmajid2021analysis}&0.540487&$\mathcal{ N}\left(0.5402435\;\; 3.7815\times10^{-4} \right)$.\\
$\beta_{3}$ & 1& (Estimated)&0.999301&$\mathcal{ N}\left(0.9996505\;\; 1.6521\times10^{-4} \right)$.\\
$\mu_1$ & 0.0142& \cite{world2010working,world2013expert}&0.014417&$\mathcal{ N}\left(0.014309\;\; 1.53\times10^{-4} \right)$.  \\
$\sigma_1$ & 1& \cite{zhang2011analysis,abdulmajid2021analysis} &1.006332&$\mathcal{ N}\left(1.03166\;\; 4.47\times10^{-3} \right)$. \\
$\theta_{2}$ & 1000& (Estimated)&1004.12044&$\mathcal{N}\left(1002.060222\;\; 2.913594\right)$. \\
$\kappa_{1}$ &  0.00006 & (Estimated)&0.000020&$\mathcal{ N}\left(0.000040\;\; 2.8\times10^{-5} \right)$.\\
$\kappa_{2}$ & 0.00005 & (Estimated)&0.000081&$\mathcal{ N}\left(0.000066\;\; 2.2\times10^{-5} \right)$.\\
$\kappa_{3}$ & $\left[0.00001 \;\; 0.00003\right]$&(Estimated)&0.000040&$\mathcal{ N}\left(0.000025\;\; 2.1\times10^{-5} \right)$.\\
$\gamma$ & $\frac{1}{6}$ & \cite{zhang2011analysis,tian2018transmission,abdulmajid2021analysis} & 0.166374&$\mathcal{ N}\left(0.166520\;\; 2.07\times10^{-4} \right)$.\\
$\sigma_2$ & 0.09 & \cite{zhang2011analysis,addo2012seir}&0.089556&$\mathcal{ N}\left(0.089778\;\; 3.14\times10^{-4} \right)$.\\
$\mu_{2}$ & 0.067 & (Estimated)&0.066268 &$\mathcal{ N}\left(0.066634\;\; 1.58\times10^{-4} \right)$.\\
$\theta_3$ & 1200& (Estimated)&1203.844461&$\mathcal{ N}\left(1201.922230\;\; 2.718444 \right)$.\\
$\psi_{1}$ &0.0004& \cite{hampson2019potential,addo2012seir}&0.000077&$\mathcal{ N}\left(0.000238\;\; 2.28\times10^{-4} \right)$.\\
$\psi_{2}$ & 0.0004 & \cite{hailemichael2022effect}&0.000066&$\mathcal{ N}\left(0.000233\;\; 2.36\times10^{-4} \right)$. \\
$\psi_{3}$ & 0.0003 & (Estimated)&0.000030&$\mathcal{ N}\left(0.0003\;\; 1.91\times10^{-4} \right)$.\\
$\mu_3$ &0.067& (Estimated)&0.080129&$\mathcal{ N}\left(0.073565\;\; 8.056\times10^{-3} \right)$. \\
$\sigma_3$ &0.08& \cite{zhang2011analysis}&0.091393 & $\mathcal{ N}\left(0.085697\;\; 8.056\times10^{-3} \right)$. \\
$\gamma_1$ & $\frac{1}{6}$ & \cite{zhang2011analysis,tian2018transmission}&0.172489&$\mathcal{ N}\left(0.169578\;\; 4.117\times10^{-3} \right)$. \\
$\gamma_2$ & 0.09 & \cite{zhang2011analysis}&0.090308&$\mathcal{ N}\left(0.090154\;\; 2.18\times10^{-4} \right)$.\\
$\gamma_3$ & 0.05&(Estimated)&0.050128&$\mathcal{ N}\left(0.050128\;\; 9.1\times10^{-5} \right)$.\\
$\nu_1$ &0.001& (Estimated)&0.001958&$\mathcal{ N}\left(0.001479\;\; 6.77\times10^{-4} \right)$.\\
$\nu_2$ &0.006& (Estimated)&0.008971&$\mathcal{ N}\left(0.007485\;\; 2.101\times10^{-3} \right)$.\\
$\nu_3$ &0.001& (Estimated)&0.005735&$\mathcal{ N}\left(0.003367\;\; 3.3348\times10^{-3} \right)$.\\
$\mu_4$  &0.08& (Estimated)&0.080625&$\mathcal{ N}\left(0.080313\;\; 4.42\times10^{-4} \right)$. \\
$\rho_{1}$ & 10& \cite{ruan2017spatiotemporal}&9.920733&$\mathcal{ N}\left(9.960366\;\; 5.605\times10^{-2} \right)$.\\
$\rho_{2}$ & 8& (Estimated)&8.116421&$\mathcal{ N}\left(8.058211\;\; 8.2322\times10^{-2} \right)$.\\
$\rho_{3}$ & 15 & (Estimated)&14.917005 &$\mathcal{ N}\left(14.958502\;\; 5.8686\times10^{-2} \right)$.\\
$C$  &0.003  (PFU)/mL & (Estimated)&0.003011&$\mathcal{ N}\left(0.003005\;\; 8.0000\times10^{-6} \right)$.
\label{T2}
\end{longtable}
\end{center}

% ----------------------

\subsection{Parameter sensitivity analysis}

For the modeling  uncertainty and parameter estimation of rabies transmission global sensitivity analysis (GSA), 
we employed LHS-PRCC to analyze uncertainty, specifically in the context of a model with 33 parameters 
(referred to as model \eqref{eqn1}). The aim is to identify the key factors influencing the transmission 
of a new infection under a specific intervention, following the approach outlined in reference \cite{world2013expert}. 
This approach involved using LHS to sample the parameters that contribute to the variable $\mathcal{R}_{0}$ 
and calculating Partial Rank Correlation Coefficients (PRCC) for these parameters, as defined in Eq. \eqref{LHS}. 
We conducted a total of 1000 simulations for each LHS run. Our parameter sampling was based on a uniform distribution, 
treating the model parameters as input variables, and $\mathcal{R}_{0}$ as the output variable. Larger absolute PRCC 
values indicate a more substantial influence of a parameter on $\mathcal{R}_{0}$. If the $p$-value exceeds 0.05, 
it is assumed that a parameter is not statistically significant for $\mathcal{R}_{0}$. Positive values denote 
positive correlations, negative values signify negative correlations, and values within the range of -0.2 to +0.2 
represent weak or negligible correlations. PRCC values greater than +0.6 or less than -0.6 indicate strong positive 
or strong negative correlations, while values within the range of $0.2 < \text{PRCC} < 0.6 $ 
or $-0.6 < \text{PRCC} <-0.2 $ represent moderate positive or moderate negative correlations, respectively.

The PRCC between parameter $y_i$ and output $x_j$, controlling 
for the effects of other variables $Z$, is then  computed as
\begin{equation}
\begin{aligned}
\text{PRCC}(y_i, x_j | Z)= \dfrac{\text{cov}(\text{X}, \text{Y})}{\sigma_X \sigma_Y} 
= \dfrac{\text{Corr}(y_i, x_j | Z)}{\sqrt{\text{Corr}(y_i, y_i | Z) \times \text{Corr}(x_j, x_j | Z)}},
\end{aligned}
\label{LHS}
\end{equation}
where $X$  is  the matrix of model outputs (e.g., disease incidence or prevalence), 
$Y$ is  the matrix of model parameters (e.g., transmission rates, recovery rates), 
$Z$ is  the matrix of other relevant variables, \(\text{cov}(\text{X}, \text{Y})\) 
is the covariance between the model output (\(X\)) and the input parameter (\(Y\)), 
and \(\sigma_X\) and \(\sigma_Y\) are the standard deviations of \(X\) and \(Y\), respectively. 

Once PRCC and LHS values are computed for each parameter, they were plotted against 
the model output to visualize their impact. This visualization helps identify which 
parameters have the most significant influence on the model output and how changes 
in those parameters affect the model behavior. Additionally, these techniques aid 
in obtaining solutions by guiding parameter estimation processes 
and improving model calibration and validation.
	
\begin{figure}[!h]
\begin{minipage}[b]{0.45\textwidth}
\includegraphics[scale=0.2]{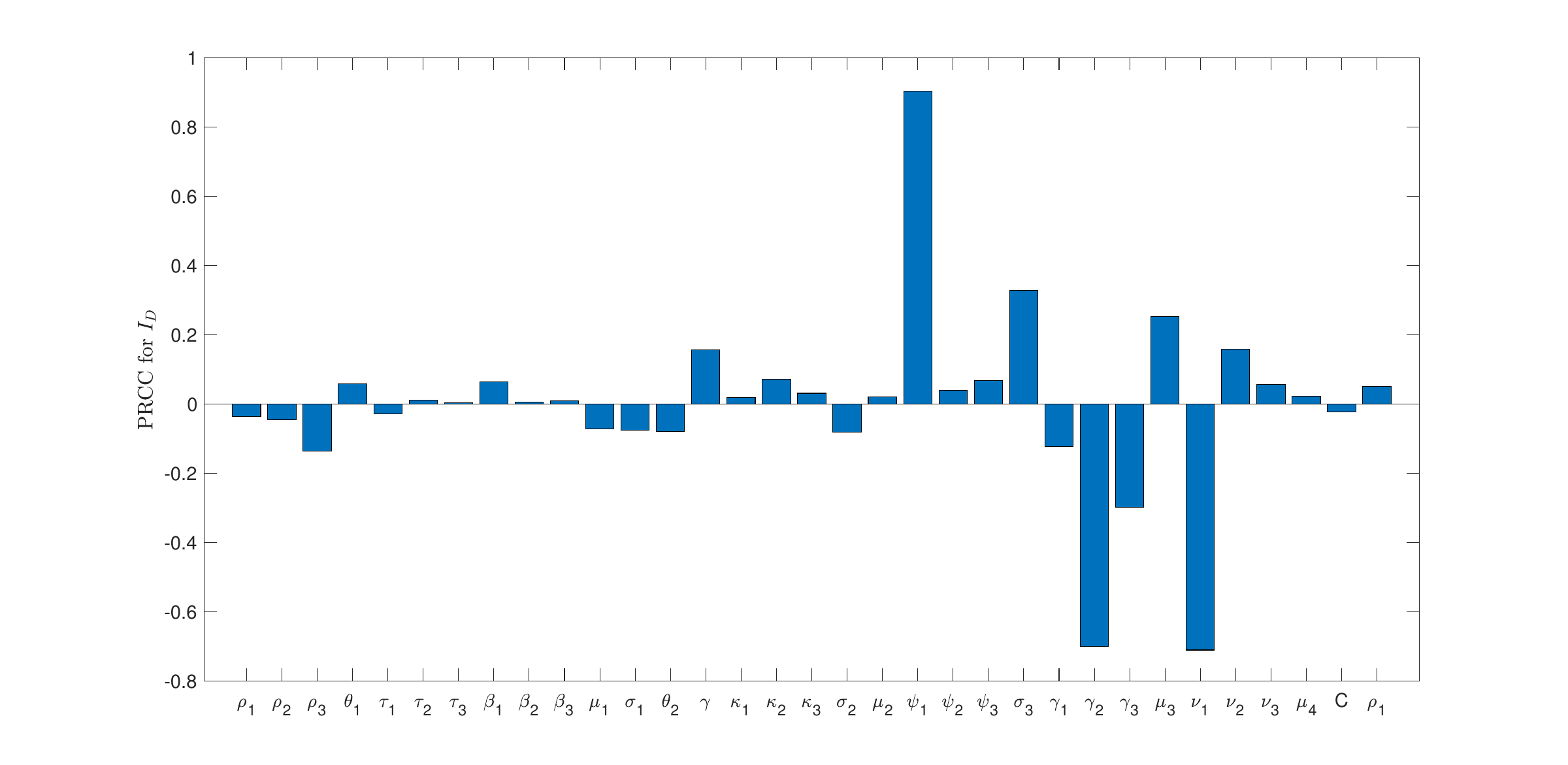}
\centering{(a)}
\end{minipage}
\begin{minipage}[b]{0.45\textwidth}
\includegraphics[scale=0.2]{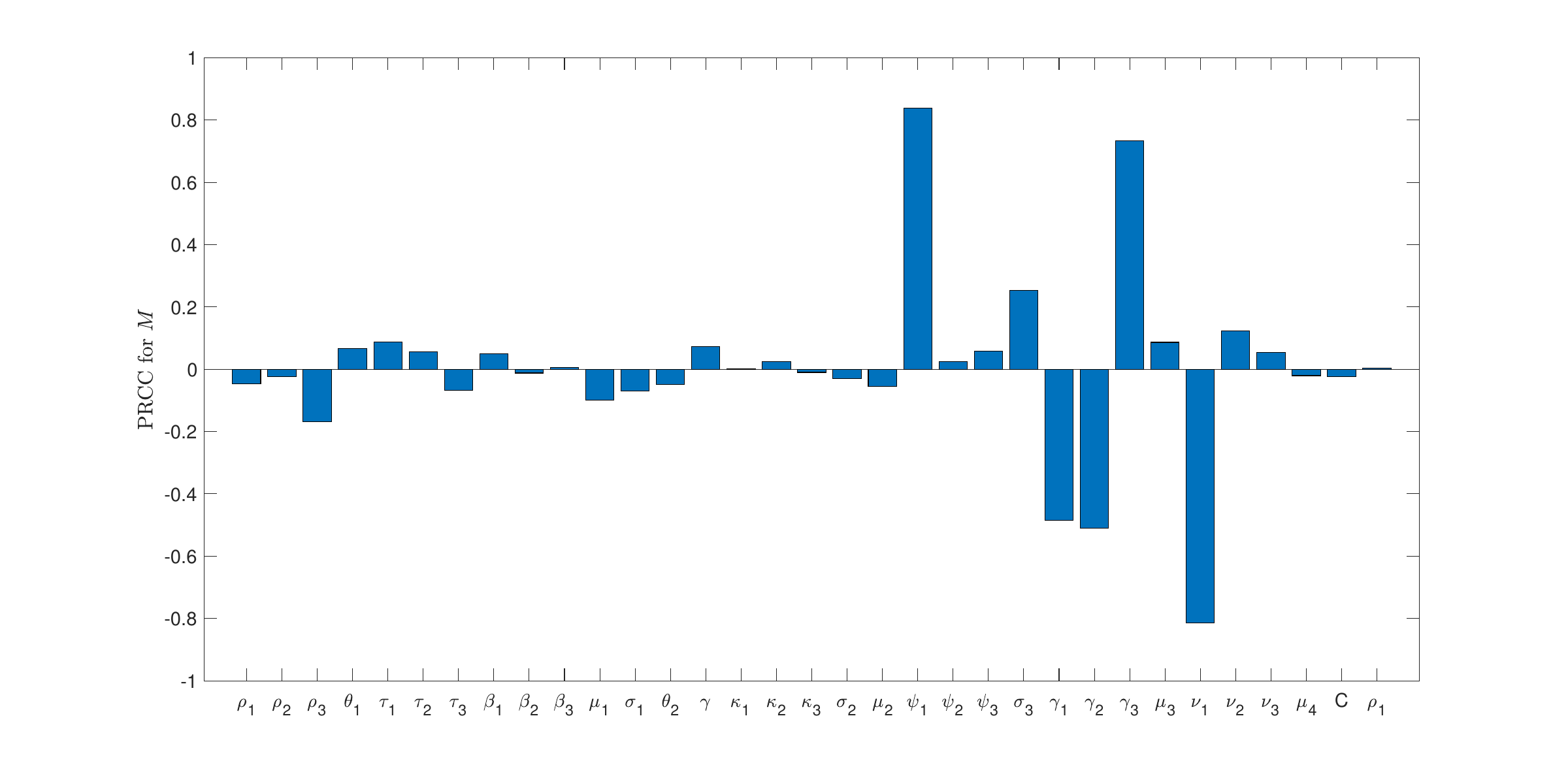}
\centering{(b)}
\end{minipage}

\begin{minipage}[b]{0.45\textwidth}
\includegraphics[scale=0.2]{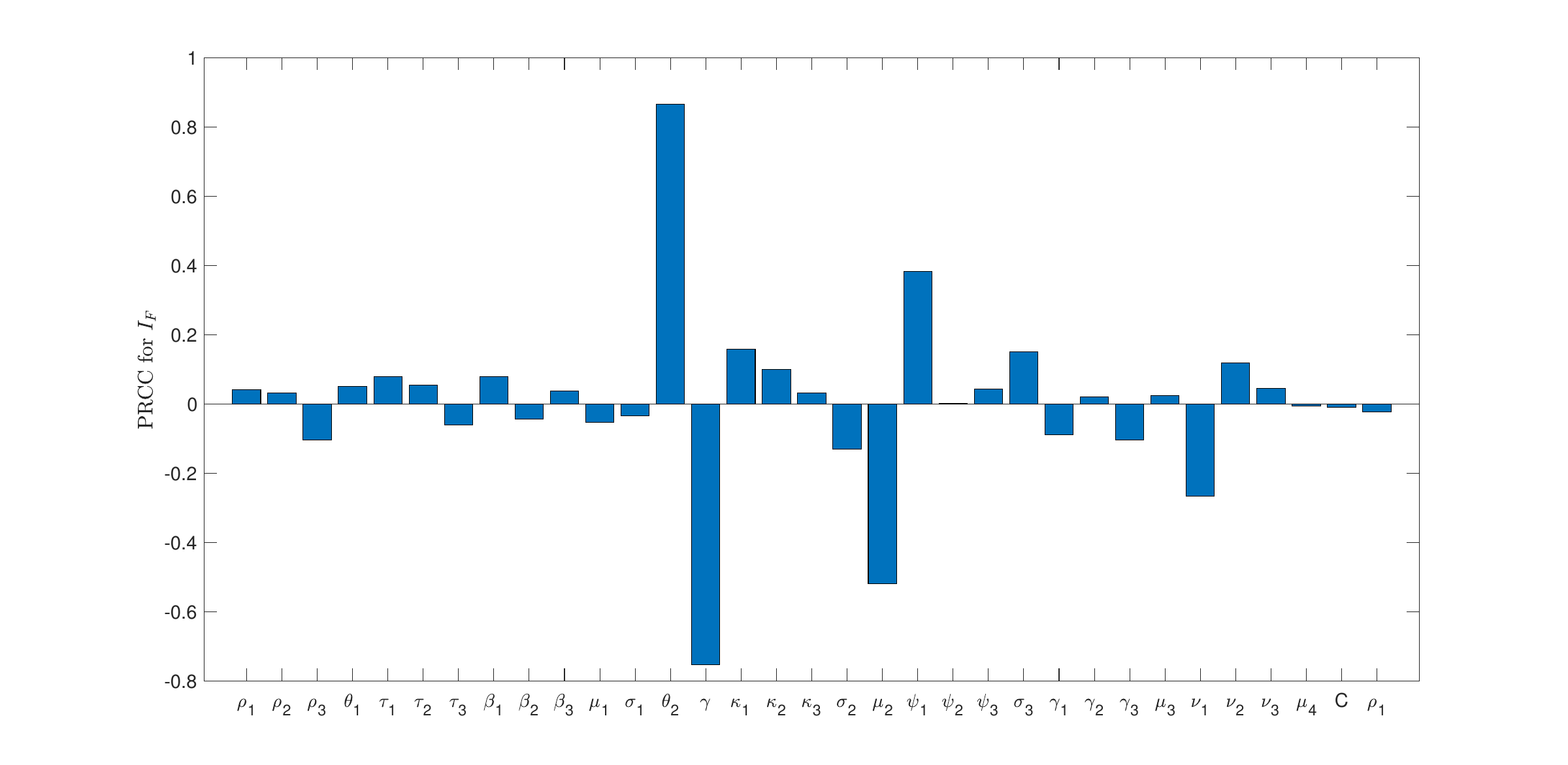}
\centering{(c)}
\end{minipage}
\begin{minipage}[b]{0.45\textwidth}
\includegraphics[scale=0.2]{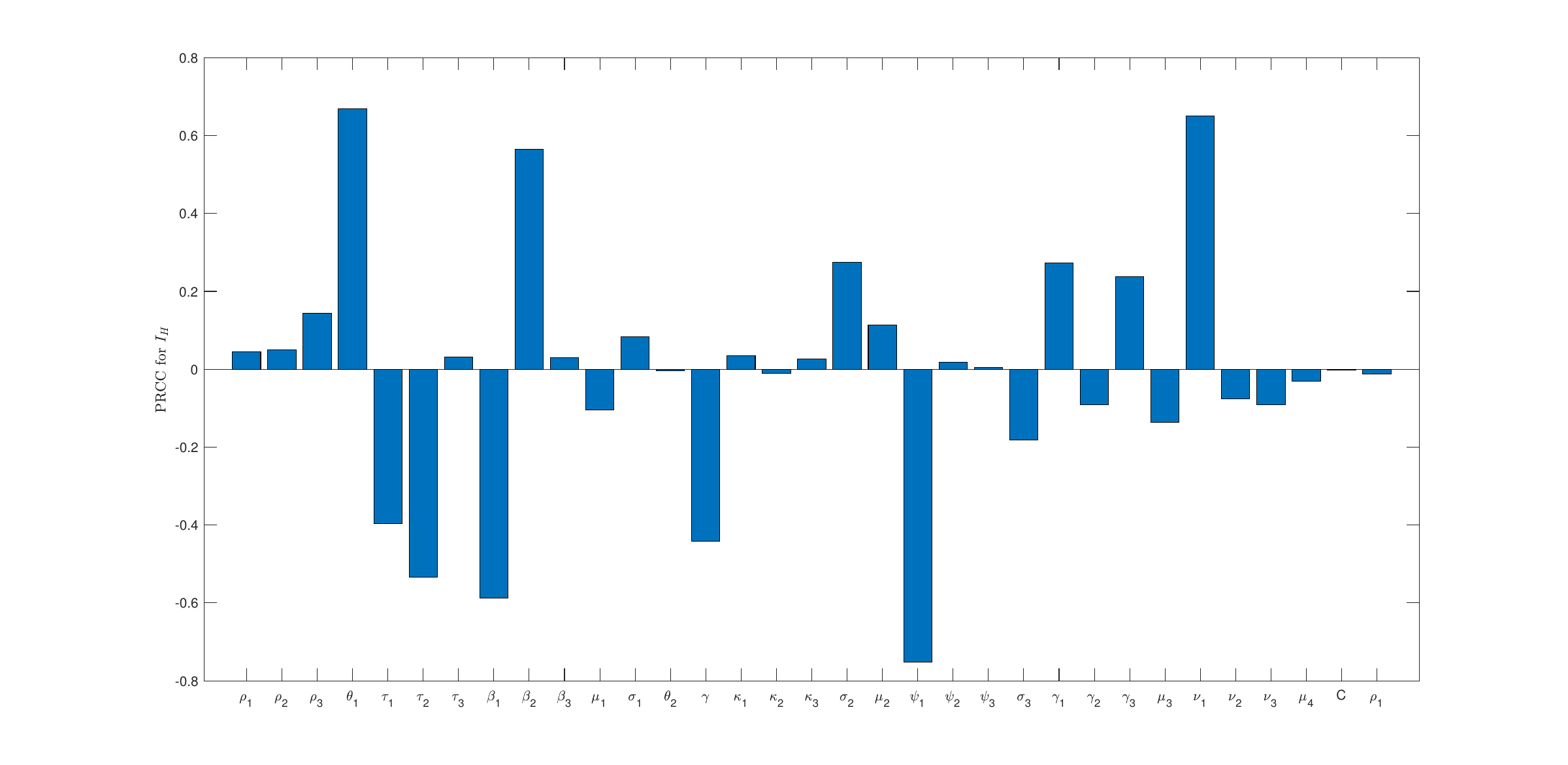}
\centering{(d)}
\end{minipage}
\caption{\centering The sensitivity analysis of  Rabies model dynamics involved 
1000 simulations employing Latin Hypercube Sampling (LHS). The analysis evaluated 
Partial Rank Correlation Coefficients (PRCCs) concerning (a) domestic dogs, (b) Environment, 
(c) Free-range dogs, and (d) Human population, respectively.}
\label{Fig7}
\end{figure}
	
\begin{figure}[H]
\begin{minipage}[b]{0.45\textwidth}
\includegraphics[scale=0.35]{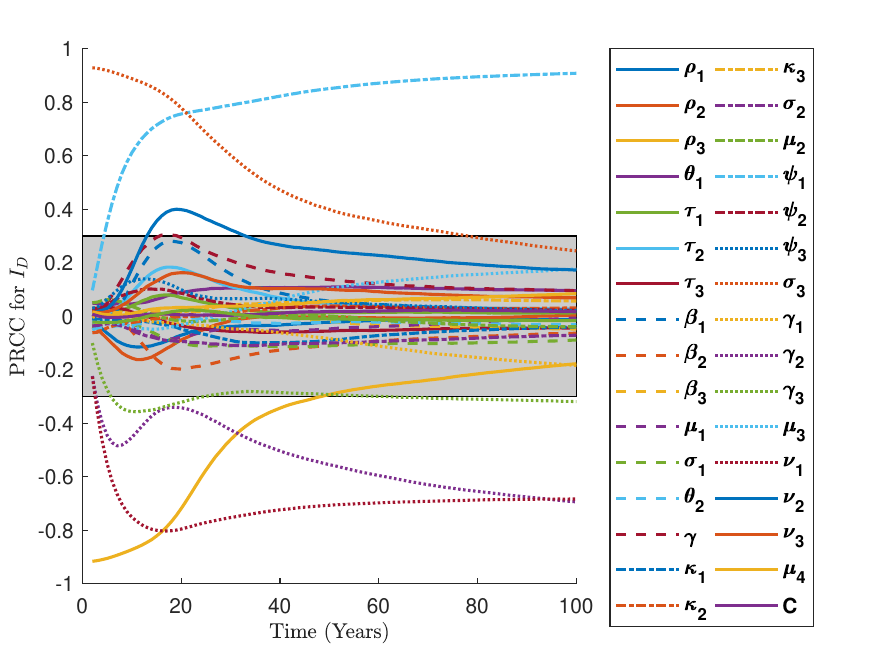}\\
\centering{(a)}
\end{minipage}
\begin{minipage}[b]{0.45\textwidth}
\includegraphics[scale=0.35]{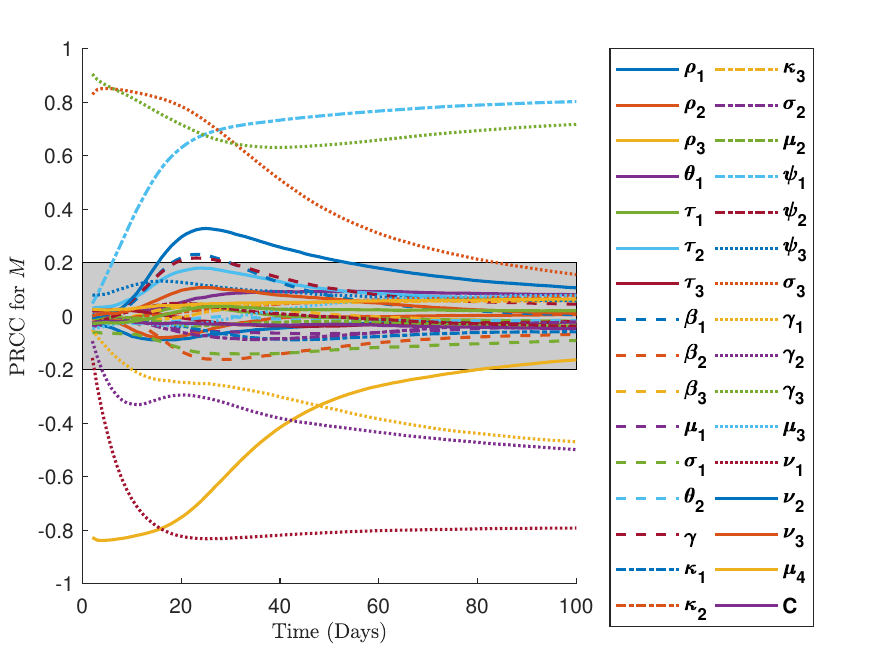}\\
\centering{(b)}
\end{minipage}

\begin{minipage}[b]{0.45\textwidth}
\includegraphics[scale=0.35]{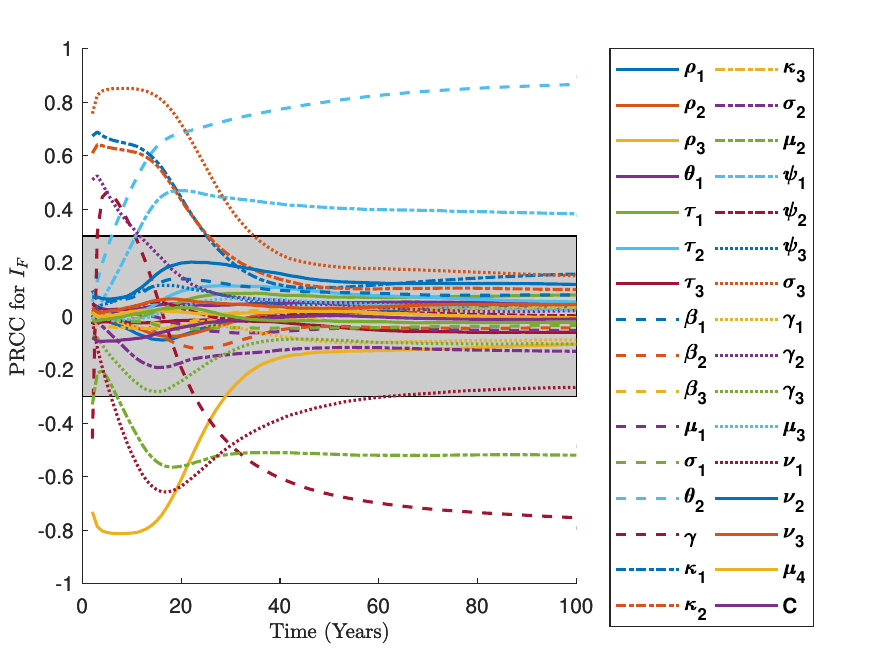}\\
\centering{(c)}
\end{minipage}
\begin{minipage}[b]{0.45\textwidth}
\includegraphics[scale=0.35]{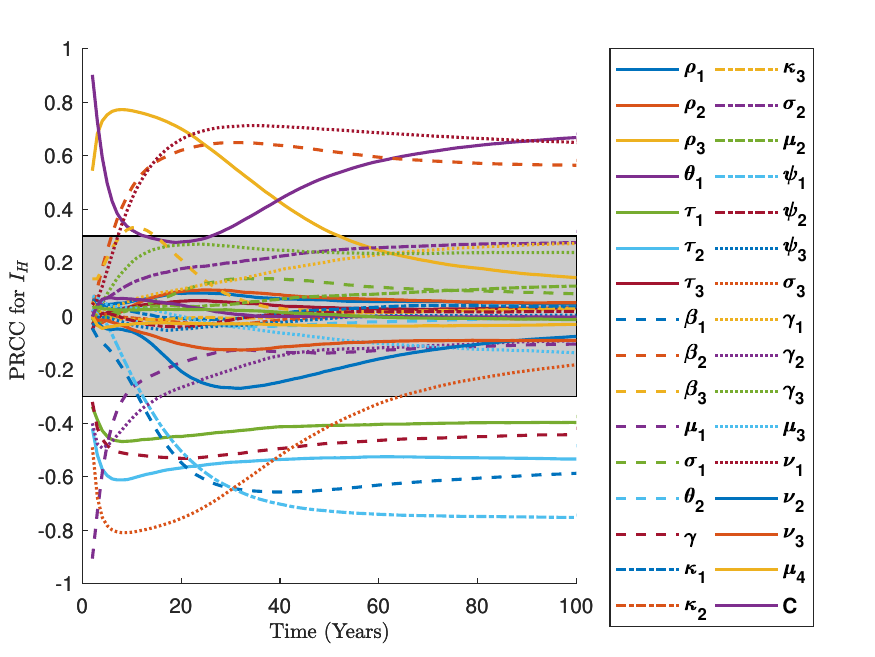}\\
\centering{(d)}
\end{minipage}
\caption{\centering A visual representation showcasing the evolution of parameter 
sensitivity throughout the progression of the system dynamics. PRRC values over time 
span of 100 years with respect to (a) domestic dogs, (b) Environment, 
(c) Free-range dogs, and (d) Human population.}
\label{Fig8}
\end{figure}
The PRCC results provide insights into the impact of each parameter on the model 
and its associated uncertainty. Parameters with PRCC values close to zero, or equal to zero, 
are considered statistically insignificant. Fig.~\ref{Fig7}  indicates that $\psi_{1}$, 
$\psi_{2}$, $\psi_{3}$, $\tau_{1}$, $\tau_{2}$, $\tau_{3}$, $\kappa_{1}$, $\kappa_{2}$, 
and $\kappa_{3}$ have  positive correlation with the  transmission  dynamics  of rabies  
based  on the model  solution  computed  in Eq.~\eqref{LHS}, while
$\rho_{1}$, $\rho_{2}$, $\rho_{3}$, $\beta_{1}$, $\beta_{2}$, $\beta_{3}$, 
$\mu_{1}$, $\sigma_{1}$, $\sigma_{2}$, $\mu_{2}$, $\theta_{3}$, $\psi_{1}$, $\psi_{2}$, 
$\psi_{3}$, $\sigma_{3}$, $\gamma_{1}$, $\gamma_{2}$, $\gamma_{3}$, $\mu_{3}$, and $C$ 
have negative  correlation with the  transmission  dynamics  of rabies. 
Fig.~\ref{Fig8} illustrates the sensitivity of PRCC values across the entire time interval 
of the model simulation. It assesses significance and demonstrates how the sensitivity 
of each parameter influences the dynamics of the model system. This suggests that, 
to control rabies in humans and  dogs, more efforts should be directed to reduce 
the rate of infections by intervening the transmission and control of the contact rate.

% ---------------------------------------------  

\subsection{Numerical simulation}

Due to the resource-intensive nature of individually solving for 33 parameters, 
we  employed the MATLAB built-in function \textsf{fminsearch}, which leverages 
the Nelder-Mead simplex algorithm \cite{lagarias1998convergence} to identify local 
minima for the residual sum of squares as presented in Eq.~\eqref{eqn6}. Our choice 
of initial parameter values was guided by the criteria established during the qualitative 
analysis along with the initial conditions $S_H\left(0\right) = 142000 $, 
$E_H\left(0\right) = 40$, $I_H\left(0\right) = 0$, $R_H\left(0\right) = 0$,  
$S_D\left(0\right) = 15000$, $E_D\left(0\right) = 25$,$I_D\left(0\right) = 0$, 
$R_D\left(0\right) = 0$, $S_F\left(0\right) = 12500$, $E_F\left(0\right) = 20$, 
$I_F\left(0\right) = 0$, $M\left(0\right) = 90$ . The estimated parameter values 
in Table~\ref{T2} were then applied to model the data, and the outcomes 
are presented in Figures~\ref{fig:HP}--\ref{fig:DD}.

\begin{figure}[H]
\centering
\includegraphics[scale=0.43]{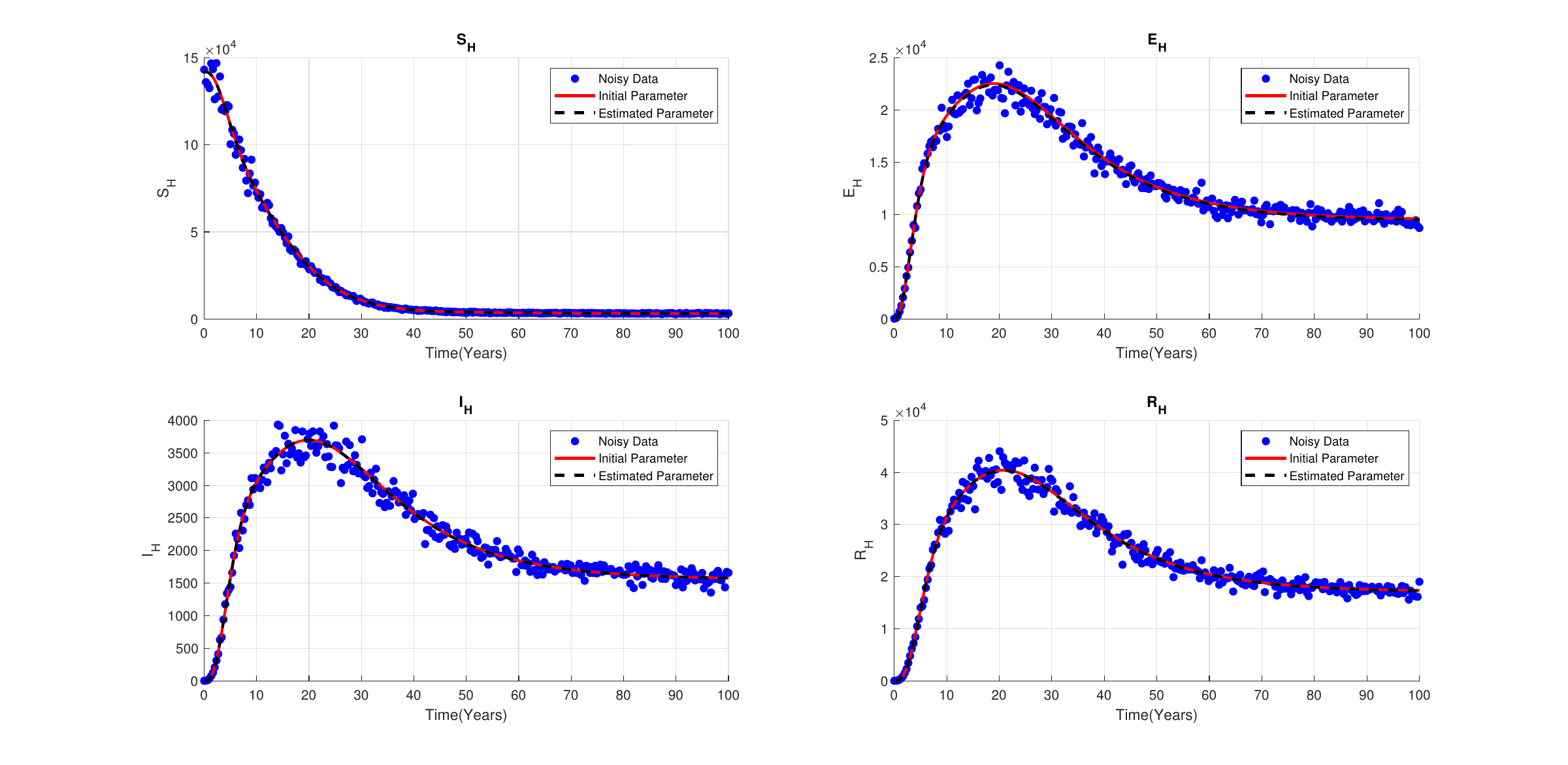}
\caption{\centering Model  fitting and  scatter plot with  corresponding  
parameter estimation   with  standard deviation $\sigma$=0.05 and  
Confidence interval (C.I.)=95\% for human population.}
\label{fig:HP}
\end{figure}
\begin{figure}[H]
\centering
\includegraphics[scale=0.55]{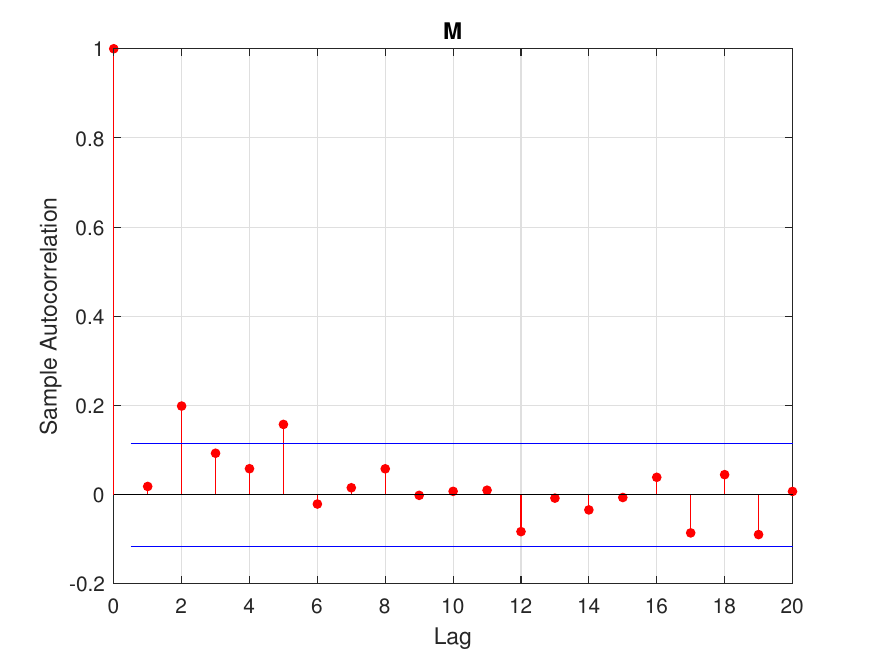}
\caption{\centering The sample autocorrelation of the residuals in relation 
to environment indicating the lack of significance at the 5\% level.}
\label{fig:AM}
\end{figure}
\begin{figure}[H]
\centering
\includegraphics[scale=0.4]{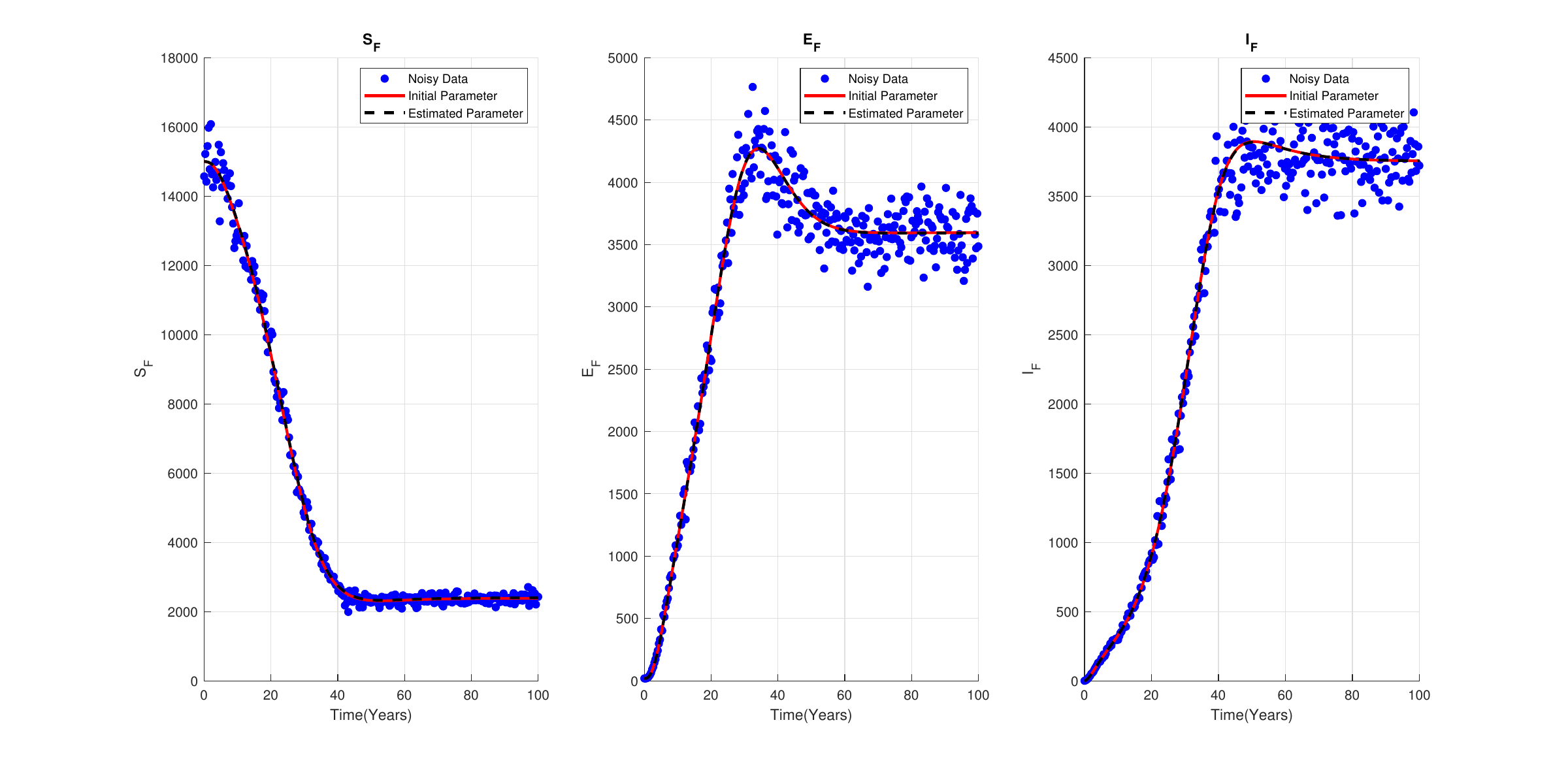}
\caption{\centering Model  fitting and  scatter plot with  corresponding  
parameter estimation   with  standard deviation $\sigma$=0.05 and  
Confidence interval (C.I.)=95\% for free range dogs.}
\label{fig:FRD}
\end{figure}
\begin{figure}[H]
\centering
\includegraphics[scale=0.40]{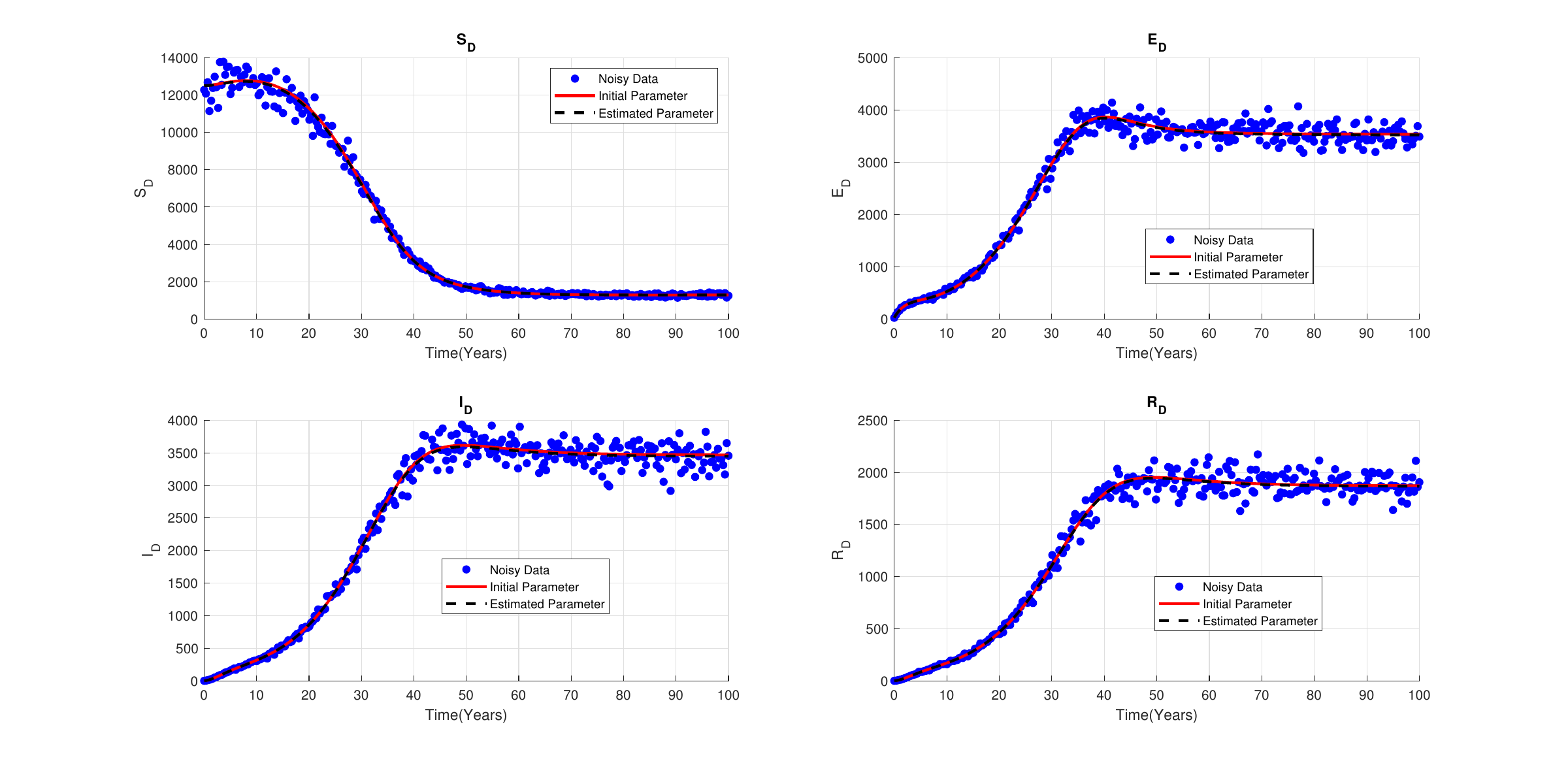}
\caption{\centering Model  fitting and  scatter plot with  corresponding  
parameter estimation   with  standard deviation $\sigma$=0.05 and  
Confidence interval (C.I.)=95\% for domestic dogs.}
\label{fig:DD}
\end{figure}

\begin{figure}[H]
\centering
\includegraphics[scale=0.65]{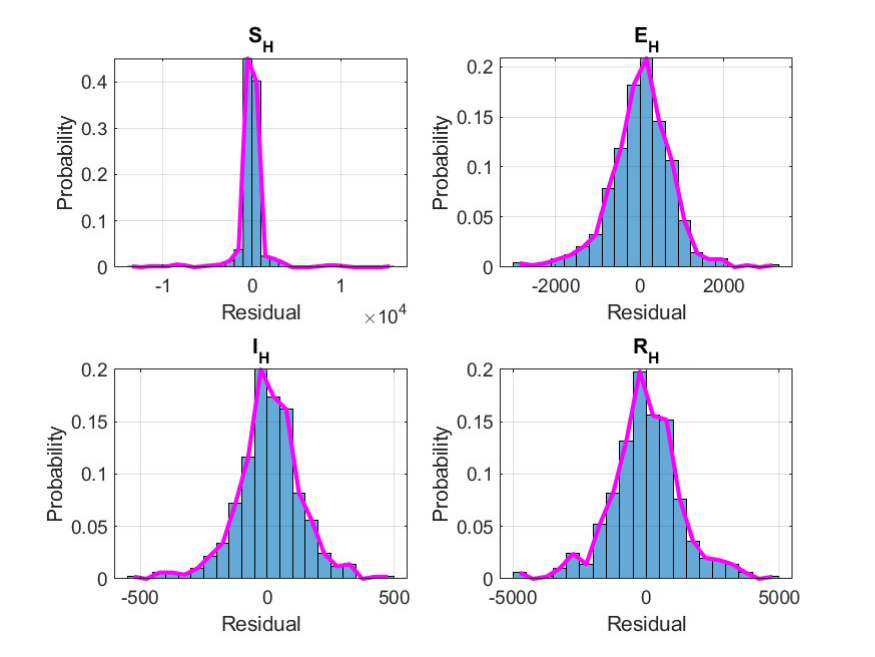}
\caption{\centering Normal distribution  of the rabies  model variable   
with  standard deviation $\sigma$=0.05 and  Confidence interval 
(C.I.)=95\% for  human population.}
\label{fig:H}
\end{figure}

\begin{figure}[H]
\centering
\includegraphics[scale=0.75]{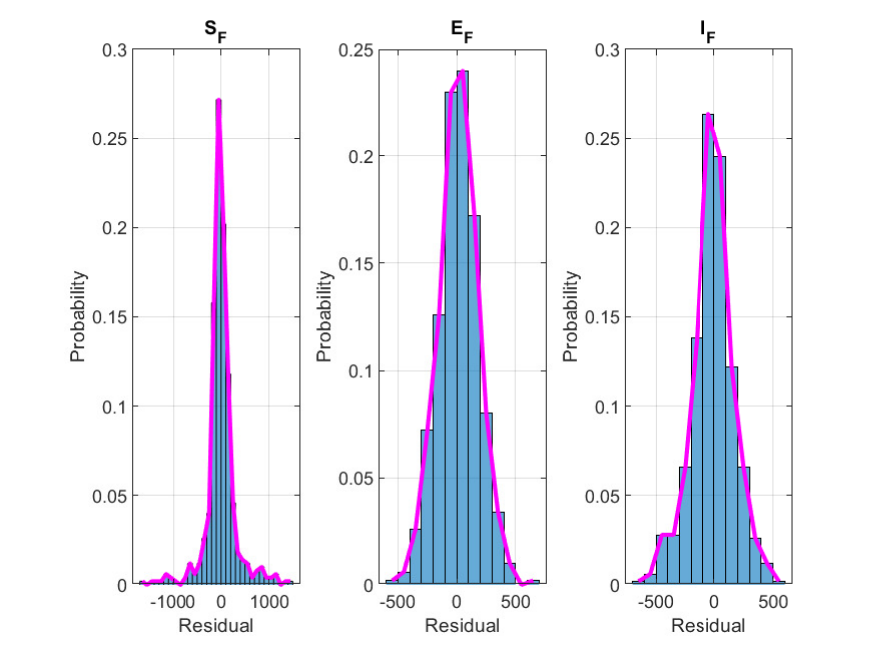}
\caption{\centering Normal distribution  of the rabies  model variable   
with  standard deviation $\sigma$=0.05 and  Confidence interval 
(C.I.)=95\% for  free-range dogs population.}
\label{fig:F}
\end{figure}
\begin{figure}[H]
\centering
\includegraphics[scale=0.75]{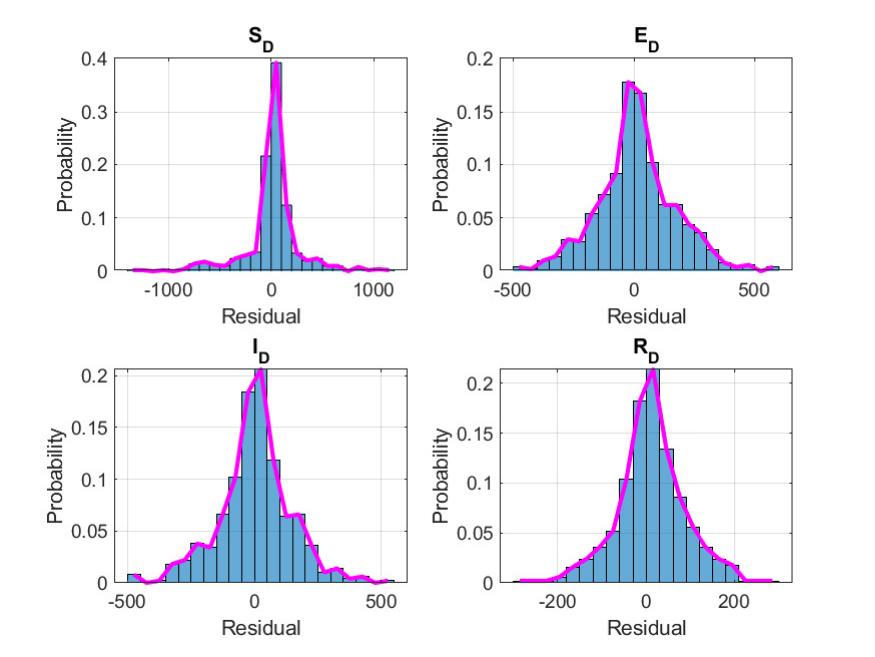}
\caption{\centering Normal distribution  of the rabies  model variable   
with  standard deviation $\sigma$=0.05 and  Confidence interval 
(C.I.)=95\% for  domestic dogs population.}
\label{fig:D}
\end{figure}

\begin{figure}[H]
\centering
\includegraphics[scale=0.75]{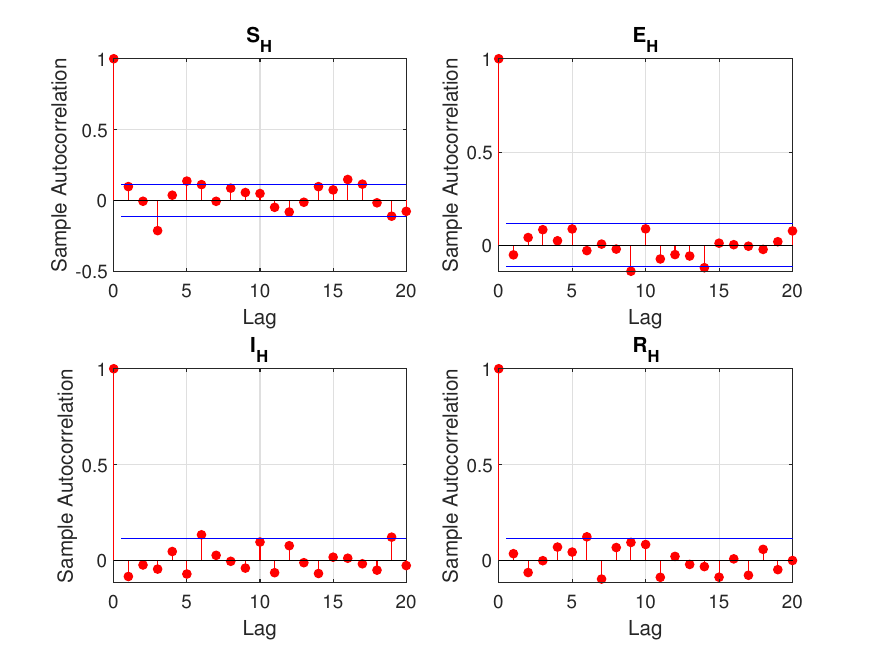}
\caption{\centering The sample autocorrelation of the residuals in relation 
to human population indicating the lack of significance at the 5\% level.}
\label{fig:AH}
\end{figure}

\begin{figure}[H]
\centering
\includegraphics[scale=0.75]{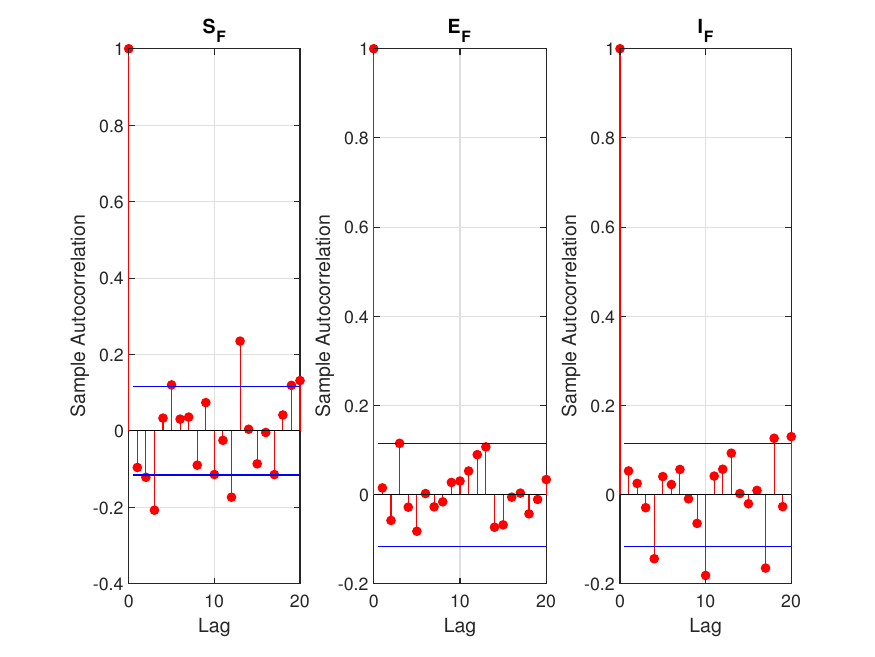}
\caption{\centering The sample autocorrelation of the residuals in relation 
to domestic dogs indicating the lack of significance at the 5\% level.}
\label{fig:AF}
\end{figure}
\begin{figure}[H]
\centering
\includegraphics[scale=0.75]{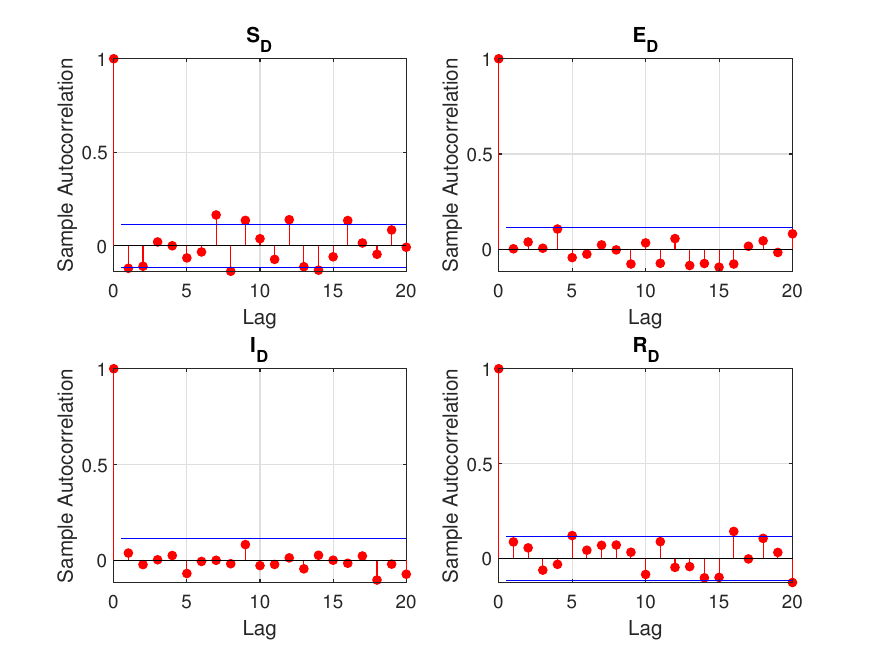}
\caption{\centering The sample autocorrelation of the residuals in relation 
to free range indicating the lack of significance at the 5\% level.}
\label{fig:AD}
\end{figure}

% ---------------------------------------------  

\section{Discussion}
\label{sec:04}

Figures~\ref{fig:H}--\ref{fig:D} illustrate the histogram of residuals, 
which provide insight into the consistency between the proposed model fitting 
and the estimated parameter values, indicating that they stem from the same probability 
distribution function. By examining the distribution of residuals, which represent 
the differences between observed data and model predictions, we can assess how effectively 
the model captures the underlying dynamics of rabies transmission. Meanwhile, 
Figs.~\ref{fig:AH}--\ref{fig:AD} show an autocorrelation between 0.2 and -0.2, 
implying a moderate to weak correlation between consecutive residuals. 
This suggests that the errors in the model predictions are relatively independent 
of each other over time, which can be interpreted as a desirable property indicating 
that the model adequately captures the underlying dynamics without systematic 
biases or trends in the residuals.

However, the spread of infectious diseases among human and dog populations 
often results in a dynamic interplay between the number of susceptible individuals 
and the number of infected individuals. In particular, Figs.~\ref{fig:HP}--\ref{fig:AD}   
show trends for humans, free-range and domestic dog populations. The data presented 
in these figures suggest that as the number of infections increases, the number 
of susceptible individuals decreases towards a non-negative constant value. 
This phenomenon can be explained by the increased contact rate between infectious 
and susceptible compartments.

% ---------------------------------------------  

\section{Conclusion} 
\label{sec:05}

In this study, we developed and analyzed a deterministic model to estimate parameters and assess 
uncertainty in the transmission dynamics of rabies in humans and dogs. Data were generated by adding 
random Gaussian noise to a real world data, and parameter values were estimated using the non-linear 
least squares method. The basic reproduction number ${\cal R}_0$ was derived using the next-generation 
matrix approach. We determined model equilibria and found that the rabies-free equilibrium is globally 
asymptotically stable when ${\cal R}_0$ is less than one, and the rabies-endemic equilibrium is globally 
stable when ${\cal R}_0$ is greater than or equal to one. Latin Hypercube and Partial Rank Correlation 
Coefficient analyses were employed to identify which parameters positively or negatively affect model outputs. 
Sensitivity analysis revealed that contact rates between humans and dogs contribute to an increase 
in average new rabies infections in the entire population.

% ---------------------------------------------  

\section*{CRediT Authorship Contribution Statement}

Mfano Charles: Writing -- original draft, Software, Methodology, Formal analysis, Conceptualization.
Sayoki G. Mfinanga: Writing -- review \& editing, Supervision, Methodology, Data curation, Conceptualization.
G.~A.~Lyakurwa: Writing -- review \& editing, Supervision.
Delfim F.~M.~Torres: Writing -- review \& editing, Supervision, Methodology.
Verdiana G.~Masanja: Writing -- review \& editing.

% ---------------------------------------------  
  
\section*{Funding Statement}

Torres is supported by the Portuguese Foundation for Science and Technology (FCT) and CIDMA, 
project UIDB/04106/2020.

% ---------------------------------------------  

\section*{Declaration of Competing Interest}

The authors declare that they have no known competing
financial interests or personal relationships that
could have appeared to influence the work reported
in this paper.

% ---------------------------------------------  

\section*{Acknowledgments}

The authors would like to extend their gratitude to the handling editor and the referee for their valuable input, 
which significantly contributed to enhancing the quality of this manuscript. Furthermore, we wish 
to acknowledge The Nelson Mandela African Institution of Science and Technology (NM-AIST) 
and the College of Business Education (CBE) for providing a conducive environment 
during the writing of this manuscript.
 
% ---------------------------------------------  

\section*{Data Availability}

The data used in this study are included in the manuscript.

% ---------------------------------------------  

% ---------------------------------------------  


\begin{thebibliography}{10}

\bibitem{bonilla2022mapping}
D.~K. Bonilla-Aldana, S.~D. Jimenez-Diaz, J.~J. Barboza, A.~J. Rodriguez-Morales, 
Mapping the spatiotemporal distribution of bovine rabies in Colombia, 2005--2019, 
Tropical Medicine and Infectious Disease 7~(12) (2022) 406.

\bibitem{barecha2017epidemiology}
C.~B. Barecha, F.~Girzaw, R.~V. Kandi, M.~Pal, 
Epidemiology and public health significance of rabies, 
Perspect Clin Res 5~(1) (2017) 55--67.

\bibitem{kabeto2021rabies}
J.~M. Kabeto, Y.~T. Tewabe, A.~W. Haile, W.~Gemechu, 
Rabies: A neglected zoonotic disease and its public health concern in Ethiopia, 
Journal of Zoonotic Diseases 5~(4) (2021) 1--11.

\bibitem{kavoosian2023comparison}
S.~Kavoosian, R.~Behzadi, M.~Asouri, A.~A. Ahmadi, M.~Nasirikenari, A.~Salehi, et~al., 
Comparison of rabies cases received by the shomal Pasteur Institute
in Northern Iran: a 2-year study, 
Global Health, Epidemiology and Genomics 2023 (2023).

\bibitem{lembo2010feasibility}
T.~Lembo, K.~Hampson, M.~T. Kaare, E.~Ernest, D.~Knobel, R.~R. Kazwala, D.~T. Haydon, S.~Cleaveland, 
The feasibility of canine rabies elimination in {A}frica: 
dispelling doubts with data, PLoS Neglected Tropical Diseases 4~(2) (2010) e626.

\bibitem{rocha2017epidemiological}
S.~Rocha, S.~de~Oliveira, M.~B. Heinemann, V.~Gon{\c{c}}alves, 
Epidemiological profile of wild rabies in {B}razil (2002--2012), 
Transboundary and Emerging Diseases 64~(2) (2017) 624--633.

\bibitem{taylor2015global}
L.~H. Taylor, L.~H. Nel, 
Global epidemiology of canine rabies: past, present, and future prospects, 
Veterinary Medicine: Research and Reports (2015) 361--371.

\bibitem{marsden2006rabies}
S.~C. Marsden, C.~R. Cabanban, 
Rabies: A significant palliative care issue, Progress in Palliative Care 14~(2) (2006) 62--67.

\bibitem{monath2013vaccines}
T.~P. Monath, 
Vaccines against diseases transmitted from animals to humans: 
a one health paradigm, Vaccine 31~(46) (2013) 5321--5338.

\bibitem{maki2017oral}
J.~Maki, A.-L. Guiot, M.~Aubert, B.~Brochier, F.~Cliquet, C.~A. Hanlon,
R.~King, E.~H. Oertli, C.~E. Rupprecht, C.~Schumacher, et~al., 
Oral vaccination of wildlife using a vaccinia--rabies-glycoprotein 
recombinant virus vaccine (RABORAL V-RG$^{\text{\textregistered}}$): a global review, 
Veterinary Research 48 (2017) 1--26.

\bibitem{velasco2017successful}
A.~Velasco-Villa, L.~E. Escobar, A.~Sanchez, M.~Shi, D.~G. Streicker, N.~F.
Gallardo-Romero, F.~Vargas-Pino, V.~Gutierrez-Cedillo, I.~Damon, G.~Emerson,
Successful strategies implemented towards the elimination of canine rabies in
the western hemisphere, 
Antiviral Research 143 (2017) 1--12.

\bibitem{gossner2020prevention}
C.~M. Gossner, A.~Mailles, I.~Aznar, E.~Dimina, J.~E. Echevarr{\'\i}a, S.~L.
Feruglio, H.~Lange, F.~P. Maraglino, P.~Parodi, J.~Perevoscikovs, et~al.,
Prevention of human rabies: a challenge for the European Union and the
European Economic Area, Eurosurveillance 25~(38) (2020) 2000158.

\bibitem{abdulmajid2021analysis}
S.~Abdulmajid, A.~S. Hassan, 
Analysis of time delayed rabies model in human and dog populations with controls, 
Afrika Matematika 32~(5-6) (2021) 1067--1085.

\bibitem{amoako2021rabies}
Y.~Amoako, P.~El-Duah, A.~Sylverken, M.~Owusu, R.~Yeboah, R.~Gorman, T.~Adade,
J.~Bonney, W.~Tasiame, K.~Nyarko-Jectey, et~al., 
Rabies is still a fatal but neglected disease: a case report, 
Journal of Medical Case Reports 15~(1) (2021) 1--6.

\bibitem{abrahamian2022rhabdovirus}
F.~M. Abrahamian, C.~E. Rupprecht, Rhabdovirus: Rabies. 
In: Viral Infections of Humans: Epidemiology and Control, 
Springer, 2022, pp. 1--49.

\bibitem{ruan2017modeling}
S.~Ruan, 
Modeling the transmission dynamics and control of rabies in {C}hina, 
Mathematical Biosciences 286 (2017) 65--93.

\bibitem{ayoade2019saturated}
A.~A. Ayoade, O.~J. Peter, T.~A. Ayoola, S Amadiegwu, A.~Victor, 
A saturated treatment model for the transmission dynamics of rabies,
Malaysian Journal of Computing 4~(1) (2019) 201--213.

\bibitem{chapwanya2022environment}
M.~Chapwanya, P.~Dumani, 
Environment considerations on the spread of rabies among {A}frican 
wild dogs (lycaon pictus) with control measures, 
Mathematical Methods in the Applied Sciences 45~(8) (2022) 4124--4139.

\bibitem{kadowaki2018risk}
H.~Kadowaki, K.~Hampson, K.~Tojinbara, A.~Yamada, K.~Makita, 
The risk of rabies spread in {J}apan: a mathematical modelling assessment, 
Epidemiology \& Infection 146~(10) (2018) 1245--1252.

\bibitem{ayoade2023modeling}
A.~A. Ayoade, M.~O. Ibrahim, 
Modeling the dynamics and control of rabies in dog population within and around {L}agos, {N}igeria, 
The European Physical Journal Plus 138~(5) (2023) 397.

\bibitem{tulu2017impact}
A.~M. Tulu, P.~R. Koya, 
The impact of infective immigrants on the spread of dog rabies, 
American Journal of Applied Mathematics 5~(3) (2017) 68.

\bibitem{pantha2021modeling}
B.~Pantha, S.~Giri, H.~R. Joshi, N.~K. Vaidya, 
Modeling transmission dynamics of rabies in {N}epal, 
Infectious Disease Modelling 6 (2021) 284--301.

\bibitem{eze2020mathematical}
O.~C. Eze, G.~E. Mbah, D.~U. Nnaji, N.~E. Onyiaji, 
Mathematical modelling of transmission dynamics of rabies virus, 
International Journal of Mathematics Trends and Technology 66 (2020).

\bibitem{ega2015sensitivity}
T.~T. Ega, L.~Luboobi, D.~Kuznetsov, A.~H. Kidane, 
Sensitivity analysis and numerical simulations for the mathematical model 
of rabies in human and animal within and around {A}ddis {A}baba, 
Asian Journal of Mathematics and Applications (2015).

\bibitem{renardy2019global}
M.~Renardy, C.~Hult, S.~Evans, J.~J. Linderman, D.~E. Kirschner, 
Global sensitivity analysis of biological multiscale models, 
Current Opinion in Biomedical Engineering 11 (2019) 109--116.

\bibitem{gomero2012latin}
B.~Gomero, 
Latin hypercube sampling and partial rank correlation 
coefficient analysis applied to an optimal control problem, 
Ph.D. thesis, University of Tennessee (2012).
\url{https://trace.tennessee.edu/utk_gradthes/1278}

\bibitem{lasalle1976stability}
J.~P. LaSalle, 
Stability theory and invariance principles. 
In: Dynamical Systems, Elsevier, 1976, pp. 211--222.

\bibitem{yang2014basic}
H.~M. Yang, 
The basic reproduction number obtained from Jacobian and next generation 
matrices--a case study of dengue transmission modelling,
Biosystems 126 (2014) 52--75.

\bibitem{saha2021dynamics}
S.~Saha, G.~Samanta, 
Dynamics of an epidemic model under the influence of environmental stress, 
Mathematical Biology and Bioinformatics 16~(2) (2021) 201--243.

\bibitem{dharmaratne2020estimation}
S.~Dharmaratne, S.~Sudaraka, I.~Abeyagunawardena, K.~Manchanayake,
M.~Kothalawala, W.~Gunathunga, 
Estimation of the basic reproduction number ($R_0$) for the novel 
coronavirus disease in {S}ri {L}anka, 
Virology Journal 17 (2020) 1--7.

\bibitem{castillo2002computation}
C.~Castillo-Chavez, Z. Feng, W. Huang,
On the computation of r. and its role on global stability,
Mathematical approaches for emerging and reemerging infectious diseases: an introduction,
1 (2002) 229.

\bibitem{kamgang2008computation}
J.~C. Kamgang, G.~Sallet, 
Computation of threshold conditions for epidemiological models 
and global stability of the disease-free equilibrium ({DFE}), 
Mathematical Biosciences 213~(1) (2008) 1--12.

\bibitem{tian2018transmission}
H.~Tian, Y.~Feng, B.~Vrancken, B.~Cazelles, H.~Tan, M.~S. Gill, Q.~Yang, Y.~Li,
W.~Yang, Y.~Zhang, et~al., 
Transmission dynamics of re-emerging rabies in domestic dogs of rural China, 
PLoS Pathogens 14~(12) (2018) e1007392.

\bibitem{zhang2011analysis}
J.~Zhang, Z.~Jin, G.-Q. Sun, T.~Zhou, S.~Ruan, 
Analysis of rabies in China: transmission dynamics and control, 
PLoS One 6~(7) (2011) e20891.

\bibitem{world2010working}
W.~H. Organization, et~al., 
Working to overcome the global impact of neglected tropical diseases: 
first WHO report on neglected tropical diseases, no.~WHO/HTM/NTD/2010.1 
in WHO Library Cataloguing-in-Publication Data, 
World Health Organization, 2010.

\bibitem{world2013expert}
W.~H. Organization, 
WHO expert consultation on rabies: second report, 
Vol. 982, World Health Organization, 2013.

\bibitem{addo2012seir}
K.~M. Addo, 
An {SEIR} mathematical model for dog rabies; case study: {B}ongo district, {G}hana, 
Ph.D. thesis, Kwame Nkrumah University of Science and Technology, Kumasi (2012).

\bibitem{hampson2019potential}
K.~Hampson, F.~Ventura, R.~Steenson, R.~Mancy, C.~Trotter, L.~Cooper,
B.~Abela-Ridder, L.~Knopf, M.~Ringenier, T.~Tenzin, et~al., 
The potential effect of improved provision of rabies post-exposure prophylaxis in
gavi-eligible countries: a modelling study, 
The Lancet Infectious Diseases 19~(1) (2019) 102--111.

\bibitem{hailemichael2022effect}
D.~D. Hailemichael, G.~K. Edessa, P.~R. Koya, 
Effect of vaccination and culling on the dynamics 
of rabies transmission from stray dogs to domestic dogs,
Journal of Applied Mathematics 2022 (2022).

\bibitem{ruan2017spatiotemporal}
S.~Ruan, 
Spatiotemporal epidemic models for rabies among animals, 
Infectious Disease Modelling 2~(3) (2017) 277--287.

\bibitem{lagarias1998convergence}
J.~C. Lagarias, J.~A. Reeds, M.~H. Wright, P.~E. Wright, 
Convergence properties of the nelder--mead simplex method in low dimensions, 
SIAM Journal on Optimization 9~(1) (1998) 112--147.

\end{thebibliography}
\end{document}